\documentclass[12pt]{amsart}

\usepackage{amsmath}
\usepackage{latexsym}
\usepackage{amsfonts}
\usepackage{amssymb}
\usepackage{mathtools}
\usepackage{color}
\usepackage{bbm,dsfont}
\usepackage{graphicx}
\usepackage{ifthen}
\usepackage{enumerate}

\newtheorem{proposition}{Proposition}
\newtheorem{theorem}{Theorem}

\newtheorem{corollary}{Corollary}

\theoremstyle{definition}

\newtheorem{remark}{Remark}
\newtheorem{example}{Example}
\newtheorem{definition}{Definition}

\newcommand{\R}{\mathbb{R}} 
\newcommand{\C}{\mathbb{C}} 
\newcommand{\F}{\mathbb{F}} 

\newcommand{\Z}{\mathbb Z} 
\newcommand{\T}{\mathbb T} 

\newcommand{\Tr}[1]{{\rm Tr}\, #1} 

\newcommand{\e}{{\rm e}}

\renewcommand{\aa}{\mathcal{A}}
\newcommand{\bb}{\mathcal{B}}

\newcommand{\oo}{\mathcal{O}}
\newcommand{\mm}{\mathcal{M}}

\newcommand{\qq}{\mathcal{Q}}

\newcommand{\yy}{\mathcal{Y}}

\newcommand{\lf}{\mathfrak{l}}
\newcommand{\mf}{\mathfrak{m}}

\newcommand{\SL}{{\rm SL}(V)}
\newcommand{\Sym}{{\rm Sym}(V)}

\newcommand{\hh}{\mathcal{H}} 
\newcommand{\lh}{\mathcal{L(H)}} 
\newcommand{\elle}[1]{\mathcal{L}(#1)} 
\newcommand{\uelle}[1]{\mathcal{U}(#1)} 
\newcommand{\ip}[2]{\left\langle\,#1\,|\,#2\,\right\rangle} 
\newcommand{\sym}[2]{S\left(#1 , #2\right)} 
\newcommand{\dual}[2]{b_S\left(#1 , #2\right)} 
\newcommand{\tr}[1]{{\rm tr}\left[#1\right]} 
\newcommand{\ran}{\textrm{ran}\,} 
\newcommand{\id}{\mathbbm{1}} 

\newcommand{\lam}{\lambda}
\newcommand{\eps}{\varepsilon}

\newcommand{\ff}{\mathcal{F}}

\newcommand{\vd}{\mathbf{d}} 
\newcommand{\ve}{\mathbf{e}} 
\newcommand{\vu}{\mathbf{u}} 
\newcommand{\vv}{\mathbf{v}} 
\newcommand{\vw}{\mathbf{w}} 
\newcommand{\vx}{\mathbf{x}} 
\newcommand{\vy}{\mathbf{y}} 
\newcommand{\vnull}{\mathbf{0}}

\newcommand{\Po}{\mathsf{Q}}
\newcommand{\Qo}{\mathsf{Q}}

\renewcommand{\ss}{\mathcal{D}}

\newcommand{\Aff}[1]{L(#1)} 
\newcommand{\Af}[2]{L_{#1}(#2)} 

\newboolean{authors}
\setboolean{authors}{true}

\begin{document}\setlength{\arraycolsep}{2pt}

\title[Covariant mutually unbiased bases]{Covariant mutually unbiased bases}

\author[Carmeli]{Claudio Carmeli}
\thanks{Claudio Carmeli, D.I.M.E., Universit\`a di Genova, Via Magliotto 2, I-17100  Savona, Italy \\
email: claudio.carmeli@gmail.com}

\author[Schultz]{Jussi Schultz}
\thanks{Jussi Schultz, Dipartimento di Matematica, Politecnico di Milano, Piazza Leonardo da Vinci 32, I-20133 Milano, Italy, and Turku Centre for Quantum Physics, Department of Physics and Astronomy, University of Turku, FI-20014 Turku, Finland\\
email: jussi.schultz@gmail.com}

\author[Toigo]{Alessandro Toigo}
\thanks{Alessandro Toigo, Dipartimento di Matematica, Politecnico di Milano, Piazza Leonardo da Vinci 32, I-20133 Milano, Italy, and I.N.F.N., Sezione di Milano, Via Celoria 16, I-20133 Milano, Italy\\
email: alessandro.toigo@polimi.it}

\date{\today}

\begin{abstract}
The connection between maximal sets of mutually unbiased bases (MUBs) in a prime-power dimensional Hilbert space and finite phase-space geometries is well known. In this article we classify MUBs according to their degree of covariance with respect to the natural symmetries of a finite phase-space, which are the group of its affine symplectic transformations. We prove that there exist maximal sets of MUBs that are covariant with respect to the full group only in odd prime-power dimensional spaces, and in this case their equivalence class is actually unique. Despite this limitation, we show that in dimension $2^r$ covariance can still be achieved by restricting to proper subgroups of the symplectic group, that constitute the finite analogues of the oscillator group. For these subgroups, we explicitly construct the unitary operators yielding the covariance.
\end{abstract}

\maketitle

\section{Introduction}

As already outlined in the seminal work of Schwinger \cite{Schwinger60} and later clarified by Bandyopadhyay, Boykin and Roychowdhury \cite{BaBoRoVa02}, the construction of mutually unbiased bases (MUBs) is closely related to the representation theory of finite Heisenberg groups. This connection explains the considerable interest that MUBs have raised among the mathematical community in recent times, and which has been further strengthened by the wealth of symmetry structures involved in this topic \cite{Ho05}. It is well known that there exists a striking analogy between the construction of a maximal set of MUBs in a prime-power dimensional Hilbert space and the definition of the quadrature observables in quantum homodyne tomography. This analogy originates from the possibility to extend the concept of phase-space to finite dimensional systems \cite{Wootters87,GiHoWo04}, and introduce objects like the Schr\" odinger representation, the symplectic group and its metaplectic representation also in the finite dimensional setting \cite{Vourdas04,ApBeCh08}.

Here we recall that a {\em finite phase-space} is an affine space modeled on a $2$-dimensional symplectic vector space over a finite field. Associating a finite phase-space with a prime-power dimensional quantum system simply consists in establishing a correspondence between quantum states and functions on such a space. Like in quantum homodyne tomography, this is done by means of the (finite) Wigner transform; its definition relies on two choices: (a) the selection of a maximal set of $d+1$ MUBs in the $d$-dimensional Hilbert space of the system; (b) their labeling with the affine lines of the phase-space. In this way, each basis corresponds to a set of $d$ parallel affine lines, being the finite dimensional analogue of a quadrature observable along the common direction of the lines; moreover, different MUBs are associated with sets of parallel lines having different directions, in agreement with the fact that there are exactly $d+1$ such directions in the finite phase-space \cite{Wootters87}. It is worth stressing that in (b) different labelings of the same $d+1$ MUBs can result in inequivalent definitions of the Wigner map. Therefore, the ordering of the bases actually is as relevant as their choice.

When representing quantum states as functions on the phase-space, it is important that the affine and symplectic structures of the phase-space are somehow taken into account and preserved. This is exactly the point where covariance enters the game. Indeed, the group of phase-space translations acts on the set of quantum states by means of the Schr\" odinger representation \cite{AusTol79,BaIt86,Var95,Vourdas97}; moreover, when $p$ is odd, this representation can be extended to the whole group of affine symplectic maps by means of the Weil (or metaplectic) representation \cite{We64,Ho73,Ge77,Ne02}. It is then desirable that the finite Wigner transform intertwines the combined actions of the translation and symplectic groups on the phase-space with the corresponding actions on quantum states, or, equivalently, that its associated set of ordered MUBs is {\em covariant} with respect to such group actions.

The study of the maximal sets of MUBs that are covariant with respect to the phase-space translations goes back to \cite{GiHoWo04}. In this paper, the authors considered a {\em particular} Schr\" odinger representation and classified all the equivalence classes of translation covariant MUBs associated with it. Here, equivalence is understood in the sense of equivalence under unitary conjugation, and we again stress that the ordering of the MUBs, i.e., their labeling with the phase-space lines, actually matters. The classification of \cite{GiHoWo04} is then achieved by uniquely associating a function $\Gamma : {\rm Aff}\times{\rm Aff}\times{\rm Aff}\to\C$ to each equivalence class of ordered MUBs, where ${\rm Aff}$ is the set of affine lines in the phase-space. This approach allows to determine the exact number of inequivalent translation covariant MUBs associated with the given Schr\" odinger representation; moreover, it makes clear that not all these MUBs are on the same footing, since some of them are `more symmetric' than others. Indeed, when one extends the covariance group to also include the symplectic transformations, it turns out that only a restricted set of MUBs are still covariant with respect to the enlarged symmetries. Moreover, while in the odd prime-power dimensional case it is always possible to find an equivalence class of MUBs that are covariant with respect to {\em the whole} symplectic group, it is unclear whether an analogous fact still holds for $2^r$-dimensional systems.

These considerations motivate a deeper analysis of the symmetry properties of covariant MUBs, which actually is the aim of the present paper. Our investigation will proceed in steps, as we will progressively focus on covariance with respect to larger subgroups of the whole group of affine symplectic phase-space transformations. Contrary to \cite{GiHoWo04}, we do not a priori fix any representation of the subgroup $G$ at hand, but we rather let such a representation directly arise from the symmetry properties of the MUBs under consideration. More precisely, for us a maximal set of MUBs is {\em covariant with respect to $G$} when the action of $G$ on the set of phase-space lines permutes the MUBs into equivalent ones. However, we do not make any assumption on the unitary operators yelding the equivalence.

Following \cite{GiHoWo04}, the basic symmetry we consider at the beginning of our analysis is covariance with respect to the phase-space translations. We will show that our approach allows more equivalence classes of translation covariant MUBs than the ones found in \cite{GiHoWo04}, reflecting the fact that, if the Schr\"odinger representation is not a priori fixed, inequivalent MUBs can be associated with different symplectic structures on the phase-space. However, quite surprisingly the existence of inequivalent translation covariant MUBs only relies on the possibility to permute the phase-space lines labeling each basis. Indeed, we will prove in Theorem \ref{teo:permutation2} that all phase-space translation covariant MUBs are unitarily equivalent {\em as sets of unordered bases}. This fact makes it clear that the choice of the correspondence between lines and MUBs is at the heart of any description of maximal MUBs by means of finite-phase space geometries, and in particular of the classifications made in \cite{GiHoWo04} and in the present paper. In particular, it shows that the different degrees of symmetry of covariant MUBs are only an effect of their labelings. Covariant MUBs are thus pointed out as a very special subset of the whole collection of maximal MUBs in a prime-power dimensional Hilbert space. Indeed, unitary equivalence of unordered noncovariant maximal MUBs does not hold in general \cite{Kan12}.

A fundamental tool in our analysis is a characterization of the equivalence classes of phase-space translation covariant MUBs that is alternative to the description by the functions $\Gamma$ used in \cite{GiHoWo04}. Indeed, we will prove that such classes of maximal MUBs are in a bijective correspondence with a special family of multipliers of the group of phase-space translations, which we call {\em Weyl multipliers}. The additional covariance properties of translation covariant MUBs are then directly related to the invariance properties of their associated Weyl multipliers. Studying the latter, we will be able to completely describe the classes of translation covariant MUBs that are also covariant with respect to specific subgroups of the symplectic group.

In particular, it turns out that there exist MUBs that are covariant with respect to the whole group of affine symplectic phase-space transformations if and only if the Hilbert space of the system is odd prime-power dimensional, and in this case their equivalence class is actually unique. We thus recover the analogue for maximal MUBs of a similar fact holding for covariant Wigner functions \cite{Gr06,Zh15}. Nevertheless, restricting to smaller subgroups $G$ properly containing the phase-space translations, $G$-covariant MUBs still exist even in dimension $2^r$. A particularly important instance, when $G$ is the analogue of the Euclidean group of quantum homodyne tomography, is the argument of Section \ref{sec:nonsplit}. The results there should be compared with the similar ones contained in \cite{GiHoWo04,Su07}, where however the construction of the unitary operators yielding the full $G$-covariance was somehow unclear (see Remark \ref{rem:phase} in Section \ref{sec:sympl}).
 
Now we briefly sketch the plan of the paper. Section \ref{sec:phasespace+quadsyst} introduces the $2$-dimensional affine space over a finite field, and defines the correspondence between affine lines of the space and maximal MUBs. According to the usual approach \cite{Ho05,Wootters87,GiHoWo04,Ivanovic81}, there and in the rest of the article we will view MUBs as sets of $1$-dimensional projections, which we call {\em quadrature systems} in analogy with their counterparts in quantum homodyne tomography.
In Section \ref{sec:covquadsyst}, we describe how the affine group acts on the set of all lines of the $2$-dimensional finite affine space, and we restrict our attention to quadrature systems that are covariant with respect to such an action. Section \ref{sec:Weylsyst} specializes to maximal MUBs that are covariant with respect to the group of the affine translations and introduces their associated Schr\"odinger representations, or, more precisely, the {\em Weyl systems} they generate. There we show that every translation covariant quadrature system endows the affine space with a canonical symplectic form, i.e., induces a {\em phase-space structure} on it. Through Weyl systems, the correspondence between translation covariant MUBs and Weyl multipliers is explained and studied in Sections \ref{sec:equiv} and \ref{sec:exis}. In Section \ref{sec:sympl} we enlarge the translation symmetry to include also nontrivial subgroups of the symplectic group, and establish the equivalence between the extended covariance properties of MUBs and the corresponding invariances of their associated Weyl multipliers. In Section \ref{sec:nonsplit} we concentrate on extended covariance with respect to {\em maximal nonsplit toruses} in the symplectic group, which are the analogues of the oscillator group of quantum homodyne tomography. Finally, in Section \ref{sec:spin1/2} we illustrate our results in the simplest possible example, that is, the $2$-dimensional qubit system, and show that this application already contains all the special features of the even prime-power dimensional case. Two appendices are provided at the end of the paper: Appendix \ref{app:projrep} reviews the main facts on projective representations that are needed in the paper; Appendix \ref{app:mult} provides an explicit construction of a Weyl multiplier in even prime-power dimensions.

{\bf Notations.}
The cardinality of any finite set $X$ is denoted by $|X|$.\\
In this paper, $\F$ will always be a finite field with characteristic $p$. We denote by ${\rm Tr} : \F \to \Z_p$ the trace functional of $\F$ over the cyclic field $\Z_p$ (see \cite[Section VI.5]{LanAlg} for the definition of ${\rm Tr}$). Moreover, $\F_*$ is the cyclic group of nonzero elements in $\F$ \cite[Theorem V.5.3]{LanAlg}. As usual, $\C$ is the field of complex numbers, and $\T = \{z\in\C \mid |z|=1\}$ is the group of complex phase factors.\\
By {\em Hilbert space} we always mean a finite dimensional complex Hilbert space. If $\hh$ is a Hilbert space, $\lh$ denotes the $C^*$-algebra of all linear operators on $\hh$. $\id\in\lh$ is the identity operator, and $\uelle{\hh} := \{U\in\lh \mid U^*U = \id\}$ is the group of unitary operators on $\hh$. The linear space $\lh$ becomes a Hilbert space when it is endowed with the Hilbert-Schmidt inner product $\ip{A}{B}_{HS} = \tr{AB^*}$ for all $A,B\in\lh$.

\section{Quadrature systems for finite affine spaces}\label{sec:phasespace+quadsyst}

In this section, we introduce the two main geometrical objects treated in the paper: the $2$-dimensional affine space over a finite field and the set of all its affine lines. Furthermore, we establish a correspondence between affine lines and maximal sets of MUBs in full generality, and preliminarily study the elementary properties of this correspondence.

Let $V$ be a vector space over the finite field $\F$ with $\dim_\F V = 2$. We recall that a set $\Omega$ is a $2$-dimensional {\em affine space} if it carries an action of the additive abelian group $V$ which is free and transitive. In particular, $\Omega$ is a finite set with cardinality $|\Omega|=|V|=|\F|^2$. The vector space $V$ is the {\em translation group} of $\Omega$, and we write $(\Omega,V)$ to stress the affine structure of $\Omega$. If $x\in\Omega$ and $\vu\in V$, we use the standard notation $x+\vu$ for the action of $\vu$ on $x$. Similarly, for any $x,y\in\Omega$, we denote by $\vu_{x,y}$ the unique vector in $V$ such that $x+\vu_{x,y} = y$. 

Given a $2$-dimensional affine space $(\Omega,V)$, we let $\ss$ be the set of $1$-dimensional subspaces of $V$, i.e.,
\begin{equation*}
\ss = \{ D\subset V \mid \text{$D$ is a $\F$-linear subspace and $\dim_\F D = 1$} \} \, ,
\end{equation*}
and we call each $D\in\ss$ a {\em direction} of $\Omega$.
For any $\vv\in V$, we write $\F\vv = \{\alpha\vv\mid\alpha\in\F\}$. Note that, if $\vv$ is nonzero, then $\F\vv\in\ss$; otherwise, $\F\vv=\{\vnull\}$. There is only a finite set of directions in $\ss$. Indeed, $\bigcup_{D\in\ss} D = V$ and $D_1\cap D_2 = \{\vnull\}$ if $D_1\neq D_2$, which, toghether with the fact that $|D|=|\F|$ for all $D\in\ss$, imply $|\ss|=|\F|+1$.

An {\em affine line} (or simply {\em line}) in $(\Omega,V)$ passing through $x\in\Omega$ and parallel to the direction $D\in\ss$ is the subset $x+D = \{x+\vd\mid\vd\in D\}\subset\Omega$. We write $\Aff{\Omega}$ for the collection of all affine lines in $(\Omega,V)$, and we also use the alternative notation $\lf,\mf$ etc.~for the elements of $\Aff{\Omega}$. The set $\Af{D}{\Omega} = \{x+D \mid x\in\Omega\}$ is the subset of $\Aff{\Omega}$ consisting of the lines parallel to the direction $D$. Note that the collection of subsets of parallel lines $\{\Af{D}{\Omega} \mid D\in\ss\}$ forms a partition of the whole set of lines $\Aff{\Omega}$. On the other hand, for a fixed direction $D$, the set of parallel lines $\Af{D}{\Omega}$ constitutes a partition of $\Omega$. For this reason and the equality $|x + D| = |D| = |\F|$ for all $x+D\in\Af{D}{\Omega}$, we have $|\Af{D}{\Omega}| = |\F|$. It follows that $\Aff{\Omega}$ also is finite, with $|\Aff{\Omega}| = |\ss||\Af{D}{\Omega}| = |\F|(|\F| + 1)$.

The group $V$ translates lines in $\Aff{\Omega}$ preserving their directions: if $\lf = x+D$ is a line and $\vv\in V$, we denote by $\lf+\vv = x+\vv+D$ the translate of $\lf$ by the vector $\vv$. The action of $V$ on the set of parallel lines $\Af{D}{\Omega}$ is transitive, and $D$ is the stabilizer subgroup of any line $\lf\in\Af{D}{\Omega}$; hence, the action of $V$ factors to a free and transitive action of the quotient group $V/D$ on $\Af{D}{\Omega}$. As a consequence, if $D'\neq D$ and $\lf\in\Af{D}{\Omega}$, we have $\Af{D}{\Omega} = \{\lf+\vd'\mid\vd'\in D'\}$ by the isomorphism $D'\simeq V/D$.

\begin{remark}\label{rem:explicit1}
A concrete and standard realization of the affine space $(\Omega,V)$ is obtained by setting $\Omega = V = \F^2$, that is, the set of $2$-component column arrays with entries in $\F$. The $\F$-linear vector space structure of $V$ is clear, and the action of a vector $\vv=(\alpha_1,\alpha_2)^T\in V$ on a point $x=(\gamma_1,\gamma_2)^T\in\Omega$ is by componentwise summation: $x+\vv = (\gamma_1 + \alpha_1 , \gamma_2 + \alpha_2)^T$.
The directions of $\Omega$ are
$$
\ss = \{\F(1,\alpha)^T\mid\alpha\in\F\} \cup \{\F(0,1)\}
$$
and the corresponding sets of parallel lines are
\begin{align*}
\Af{\F(1,\alpha)^T}{\Omega} & = \{\{(\lam,\beta+\lam\alpha)^T\mid\lam\in\F\} \mid \beta\in\F\} \\
\Af{\F(0,1)^T}{\Omega} & = \{\{(\beta,\lam)^T\mid\lam\in\F\} \mid \beta\in\F\} \,.
\end{align*}
\end{remark}

Now we are ready to introduce maximal sets of MUBs and associate them to our affine space $(\Omega,V)$. A convenient way to do this is by means of the projection operators on each vector of the bases, as clarified in the next definition.

\begin{definition}\label{def:quad}
A {\em quadrature system} (sometimes simply {\em quadratures}) for the affine space $(\Omega,V)$ acting on the Hilbert space $\hh$ is a map $\Po:\Aff{\Omega}\to\lh$ such that
\begin{enumerate}[(i)]
\item $\Po(\lf)$ is a rank-$1$ orthogonal projection for all $\lf\in\Aff{\Omega}$;
\item for all $D\in\ss$,
\begin{equation*}
\sum_{\lf\in\Af{D}{\Omega}} \Po(\lf) = \id \,;
\end{equation*}
\item for all $D_1,D_2\in\ss$ with $D_1\neq D_2$,
\begin{equation*}
\tr{\Po(\lf_1)\Po(\lf_2)} = \frac{1}{|\F|} \qquad \text{if $\lf_1 \in \Af{D_1}{\Omega}$ and $\lf_2 \in \Af{D_2}{\Omega}$} \,.
\end{equation*}
\end{enumerate}
\end{definition}

Note that conditions (i) and (ii) imply that the ranges of the projections $\Po(\lf_1)$ and $\Po(\lf_2)$ are orthogonal if the lines $\lf_1$ and $\lf_2$ are parallel with $\lf_1\neq\lf_2$. Since there are $|\F|$ parallel lines for each direction, this then requires that $\hh$ is a $|\F|$-dimensional Hilbert space. Picking a unit vector $\phi_\lf\in\Po(\lf)\hh$ for each line $\lf\in\Aff{\Omega}$, we also see that the set $\bb_D = \{\phi_\lf\mid\lf\in\Af{D}{\Omega}\}$ is an orthonormal basis of $\hh$ for each $D\in\ss$, and the collection of bases $\{\bb_D\mid D\in\ss\}$ is a set of $\dim\hh + 1$ MUBs by (iii). Thus, quadrature systems and maximal sets of MUBs are equivalent notions.

It is much easier to work with quadrature systems rather than directly with MUBs. As an example, Wootters and Fields proved the following very important property which will be used repeatedly in the paper.

\begin{proposition}\label{prop:reconI}
The set $\{\Po(\lf) \mid \lf\in\Aff{\Omega}\}$ spans the linear space $\lh$.
\end{proposition}
\begin{proof}
For the reader's convenience, we report the proof of \cite{WoFi89}. For all $\lf\in\Aff{\Omega}$, define the operator $Y(\lf) = \Po(\lf) - \id/|\F|$, and, for any $D\in\ss$, the set $\yy_D = \{Y(\lf)\mid\lf\in\Af{D}{\Omega}\}$. By the mutual unbiasedness condition, if $D_1,D_2\in\ss$ with $D_1\neq D_2$ and $\lf_i\in\Af{D_i}{\Omega}$, then
$$
\ip{Y(\lf_1)}{Y(\lf_2)}_{HS} = \tr{(\Po(\lf_1) - \id/|\F|)(\Po(\lf_2) - \id/|\F|} = 0 \,.
$$
That is, the two sets $\yy_{D_1}$ and $\yy_{D_2}$ are orthogonal in $\lh$. Moreover, since the sets $\{\id\}\cup\yy_D$ and $\{\Po(\lf)\mid\lf\in\Af{D}{\Omega}\}$ span the same $|\F|$-dimensional linear space, there must be at least $|\F|-1$ linearly independent operators in $\yy_D$. Actually, as $\ip{\id}{Y(\lf)}_{HS} = \tr{Y(\lf)} = 0$, there needs to be exactly $|\F|-1$ linearly independent operators in $\yy_D$. Thus, the operators $\{Y(\lf)\mid\lf\in\Aff{\Omega}\} = \bigcup_{D\in\ss}\yy_D$ span a $(|\F|+1)(|\F|-1)=(|\F|^2-1)$-dimensional space, that is, $\id^\perp$. Hence, the set $\{\id,Y(\lf)\mid\lf\in\Aff{\Omega}\}$ generates $\lh$, which implies the claim.
\end{proof}

There is a natural notion of equivalence between quadrature systems (cf.~\cite[Section VI]{GiHoWo04}).

\begin{definition}\label{def:equivP}
Two quadrature systems $\Po_1$ and $\Po_2$ for the affine space $(\Omega,V)$ acting on the Hilbert spaces $\hh_1$ and $\hh_2$, respectively, are {\em equivalent} if there exists a unitary map $U:\hh_1 \to \hh_2$ such that $\Po_2(\lf) = U\Po_1(\lf)U^*$ for all $\lf\in\Aff{\Omega}$. In this case, we write $\Po_1\sim\Po_2$ and say that $U$ {\em intertwines} $\Po_1$ with $\Po_2$.
\end{definition}

In the rest of this paper, we will be  more concerned with equivalence classes of quadrature systems rather than with their explicit realizations on specific Hilbert spaces. In particular, our focus will be on the equivalence classes that are invariant under the action of subgroups of the affine group of $(\Omega,V)$. The next section is devoted to the precise statement of our problem.

\section{The finite affine group and covariant quadrature systems}\label{sec:covquadsyst}

We have already seen that by its very definition the affine space $(\Omega,V)$ carries an action of the translation group $V$. This action can be naturally extended to the group ${\rm GL}(V)$ of all the invertible $\F$-linear maps of $V$ into itself by using the following standard procedure. First of all, one needs to choose an {\em origin} point $o\in\Omega$; once $o$ is fixed, the action is then
$$
A\cdot x = o + A\vu_{o,x} \qquad \forall x\in\Omega,\,A\in {\rm GL}(V) \,.
$$
The actions of the two groups $V$ and ${\rm GL}(V)$ combine together to yield the following action of the semidirect product ${\rm GL}(V)\rtimes V$ on $\Omega$
\begin{equation}\label{eq:actiomega}
(A,\vv) \cdot x = o + A(\vu_{o,x}+\vv) \qquad \forall x\in\Omega,\, (A,\vv)\in {\rm GL}(V)\rtimes V \,.
\end{equation}
The group ${\rm GL}(V)\rtimes V$ is the {\em affine group} of $(\Omega,V)$. Contrary to the case of the translation group, its action depends on the choice of the origin $o$, that is, the unique point of $\Omega$ such that ${\rm GL}(V)\cdot o = \{o\}$.

\begin{remark}\label{rem:explicit2}
In the concrete realization of Remark \ref{rem:explicit1}, the group ${\rm GL(V)}$ is the group of invertible $2\times 2$-matrices with entries in $\F$, which acts on $V=\F^2$ by left multiplication. The same action by left multiplication can be defined also on $\Omega=\F^2$. It corresponds to choosing the origin $o = (0,0)^T$. The overall action of the affine group ${\rm GL(V)}\rtimes V$ on $\Omega$ given in \eqref{eq:actiomega} is thus
$$
\left(\left(\begin{array}{cc} \beta_{11} & \beta_{12} \\ \beta_{21} & \beta_{22} \end{array}\right) , \left(\begin{array}{c} \alpha_1 \\ \alpha_2\end{array}\right)\right) \cdot \left(\begin{array}{c} \gamma_1 \\ \gamma_2\end{array}\right) = \left(\begin{array}{c} \beta_{11}(\gamma_1+\alpha_1) + \beta_{12}(\gamma_2+\alpha_2) \\ \beta_{21}(\gamma_1+\alpha_1) + \beta_{22}(\gamma_2+\alpha_2) \end{array}\right)
$$
for all $\alpha_i,\beta_{ij},\gamma_i\in\F$ with $\beta_{11}\beta_{22} - \beta_{12}\beta_{21} \neq 0$.
\end{remark}

Formula \eqref{eq:actiomega} can be lifted to an action of the affine group ${\rm GL}(V)\rtimes V$ on the set of the affine lines of $(\Omega,V)$. This is done by setting
$$
(A,\vv) \cdot (x+D) = (A,\vv)\cdot x + AD \qquad \forall x+D\in\Aff{\Omega},\, (A,\vv)\in {\rm GL}(V)\rtimes V .
$$
The previous definition clearly carries over to quadratures: if $\Po$ is a quadrature system for $(\Omega,V)$ acting on $\hh$ and $g\in {\rm GL}(V)\rtimes V$ is any affine transformation, then the map $\Po_g : \Aff{\Omega} \to \lh$ with
$$
\Po_g(\lf) = \Po(g\cdot \lf) \qquad \forall \lf\in\Aff{\Omega}
$$
is again a quadrature system still acting on the same Hilbert space $\hh$ of $\Po$. The relation $\Po_1\sim\Po_2$ obviously implies $\Po_{1\,g}\sim\Po_{2\,g}$. The focus of the paper will be the following special type of quadrature systems.

\begin{definition}\label{def:VcovarP}
Let $G\subseteq {\rm GL}(V)\rtimes V$ be any subgroup. A quadrature system $\Po$ for the affine space $(\Omega,V)$ is {\em $G$-covariant} if $\Po\sim\Po_g$ for all $g\in G$.
\end{definition}

We denote by $\qq_G(\Omega,V)$ the set of all $G$-covariant quadrature systems for the affine space $(\Omega,V)$. By transitivity, if $\Po\in\qq_G(\Omega,V)$ and $\Po'\sim\Po$, then also $\Po'\in\qq_G(\Omega,V)$. Clearly, $\qq_{G_2}(\Omega,V) \subseteq \qq_{G_1}(\Omega,V)$ whenever $G_1\subseteq G_2$. Moreover, if $G = \{(I,\vnull)\}$ is the one-element subgroup, then $\qq_{\{(I,\vnull)\}}(\Omega,V)$ is the set of all quadratures for the affine space $(\Omega,V)$. Our main task then will be the following:

{\em For any subgroup $G\subseteq {\rm GL}(V)\rtimes V$, completely characterize the partition of the set $\qq_G(\Omega,V)$ into equivalence classes of quadratures.}

If $\Po\in\qq_G(\Omega,V)$ acts on the Hilbert space $\hh$ and $g\in G$ is any group element, Definitions \ref{def:equivP} and \ref{def:VcovarP} imply the existence of a unitary operator $U(g)\in\lh$ such that
\begin{equation}\label{eq:defU}
\Po(g\cdot\lf) = U(g)\Po(\lf)U(g)^* \qquad \forall\lf\in\Aff{\Omega}\,.
\end{equation}
The choice of $U(g)$ is unique up to a certain extent. Indeed,

\begin{proposition}\label{prop:propW}
If $U_1(g)$ and $U_2(g)$ are two unitary operators which satisfy \eqref{eq:defU}, there exists a phase factor $a(g)\in\T$ such that $U_2(g) = a(g) U_1(g)$. Moreover, if a map $U:G\to\uelle{\hh}$ is such that \eqref{eq:defU} holds for all $g\in G$, then $U$ is a projective representation of the group $G$ in the Hilbert space $\hh$ of the quadrature system $\Po$.
\end{proposition}
\begin{proof}
Suppose both $U_1(g)$ and $U_2(g)$ satisfy \eqref{eq:defU}. Then
$$
U_2(g)^* U_1(g) \Po(\lf) = U_2(g)^* \Po(g\cdot \lf) U_1(g) =  \Po(\lf) U_2(g)^* U_1(g)
$$
for all $\lf\in\Aff{\Omega}$. Since the operators $\{\Po(\lf) \mid \lf \in\Aff{\Omega}\}$ span the whole algebra $\lh$ by Proposition \ref{prop:reconI}, we must have $U_2(g)^* U_1(g) = a(g)\id$ for some complex number $a(g)\in\T$, which yields the first claim. For the second, given $g_1,g_2\in G$, note that the unitary operators $U_1 = U(g_1 g_2)$ and $U_2 = U(g_1)U(g_2)$ both satisfy the relation $\Po((g_1 g_2)\cdot\lf) = U_i\Po(\lf)U_i^*$ for all $\lf\in\Aff{\Omega}$. Hence $U(g_1 g_2) = m(g_1,g_2) U(g_1)U(g_2)$ for some phase factor $m(g_1,g_2)\in\T$, that is, $U$ is a projective representation of $G$.
\end{proof}

We refer to Appendix \ref{app:projrep} for a brief review on projective representations. Any projective representation $U$ of $G$ which satisfies \eqref{eq:defU} will be called {\em associated} with the $G$-covariant quadrature system $\Po$. By Proposition \ref{prop:propW}, such a projective representation $U$ is uniquely determined up to multiplication by an arbitrary phase function: if $a:G\to\T$ is any map, then the projective representation $U'=aU$ also works in \eqref{eq:defU}, and there is no a priori criterion for preferring $U$ to $U'$. It is thus reasonable to try to remove this ambiguity and seek for a choice of $U$ that is canonical in some sense. In the case in which $G$ coincides with the translation group $V$, this problem will be addressed and solved in the next section.

\section{$V$-covariant quadratures and their associated Weyl systems}\label{sec:Weylsyst}

Up to now, we assumed that $(\Omega,V)$ is merely an affine space, and no further structure was postulated on it. However, we will see in Proposition \ref{prop:WHquad} below that a phase-space structure naturally arises when we restrict our analysis to maximal sets of MUBs that are covariant with respect to the group $G\equiv V$ of translations of $\Omega$.

Here we recall that the affine space $(\Omega,V)$ is a $2$-dimensional {\em phase-space} if the vector space $V$ is a {\em symplectic space}, that is, it is endowed with a symplectic form. By {\em symplectic form} we mean a nonzero $\F$-bilinear map $S : V\times V\to\F$ such that the equality $\sym{\vu}{\vu} = 0$ holds for all $\vu\in V$. The polarization identity
$$
\sym{\vu}{\vv} + \sym{\vv}{\vu} = \sym{\vu+\vv}{\vu+\vv} - \sym{\vu}{\vu} - \sym{\vv}{\vv}
$$
then implies that $S$ is antisymmetric in characteristic $p\neq 2$ and symmetric when $p=2$. Since $V$ is $2$-dimensional, $S$ is automatically nondegenerate, that is, $\sym{\vw}{\vv} = 0$ for all $\vw\in V$ only if $\vv=0$. It follows that there exists a {\em symplectic basis} $\{\ve_1,\ve_2\}$ of $V$, i.e., a linear basis of $V$ over $\F$ such that $\sym{\ve_1}{\ve_2} = - \sym{\ve_2}{\ve_1} = 1$. Moreover, all symplectic forms on $V$ only differ by a scalar factor, that is, if $S_1$ and $S_2$ are two such forms, there is $\lam\in\F_*$ for which $S_2 = \lam S_1$. In order to point out the symplectic form $S$ we are fixing on $V$, we denote by $(V,S)$ and $(\Omega,V,S)$ our symplectic spaces and phase-spaces, respectively.

\begin{remark}\label{rem:explicit3}
Continuing with the explicit realization of the affine space $(\Omega,V)$ described in Remarks \ref{rem:explicit1} and \ref{rem:explicit2}, any symplectic form $S$ on $V$ is given by
$$
\sym{(\alpha_1,\alpha_2)^T}{(\beta_1,\beta_2)^T} = \lam(\alpha_1\beta_2-\alpha_2\beta_1) \qquad \forall (\alpha_1,\alpha_2)^T,(\beta_1,\beta_2)^T\in\F^2
$$
for some choice of the scalar $\lam\in\F_*$.
\end{remark}

We continue to assume that $(\Omega,V)$ is an affine space, still without fixing any phase-space structure on it. When $G\equiv V$ is the translation group, the covariance condition \eqref{eq:defU} for a quadrature system $\Po\in\qq_V(\Omega,V)$ becomes
\begin{equation}\label{eq:defW0}
\Po(\lf+\vv) = W(\vv)\Po(\lf)W(\vv)^* \qquad \forall\lf\in\Aff{\Omega},\vv\in V \,,
\end{equation}
where $W:V\to\uelle{\hh}$ is a projective representation of the abelian group $V$ in $\hh$, uniquely determined by the quadratures $\Po$ up to multiplication by an arbitrary phase function.

The next fundamental result provides insight into the properties of the representation $W$. In particular, it shows that, through $W$, the introdution of the $V$-covariant quadrature sytem $\Po$ endows the vector space $V$ with a canonical symplectic form, unambiguously defined by $\Po$, as anticipated at the beginning of the section. The antisymmetric bicharacter appearing in the following statement is defined in Appendix \ref{app:projrep} just before Proposition \ref{prop:bichar}.

\begin{proposition}\label{prop:WHquad}
Suppose the quadrature system $\Po\in\qq_V(\Omega,V)$ acts on the Hilbert space $\hh$. Then we have the following facts.
\begin{enumerate}[(a)]
\item There exists a projective representation $W$ of $V$ associated with $\Po$ such that for all $D\in\ss$ its restriction $\left.W\right|_D$ is an ordinary representation of the additive abelian group $D$.
\item There exists a unique symplectic form $S$ on $V$ such that, if $W$ is any projective representation of $V$ associated with $\Po$, the following commutation relation holds for $W$
\begin{equation}\label{eq:comm1}
W(\vu)W(\vv) = \dual{\vu}{\vv} W(\vv)W(\vu) \qquad \forall\vu,\vv\in V
\end{equation}
where $b_S:V\times V\to\T$ is the antisymmetric bicharacter of $V$ given by
\begin{equation}\label{eq:defbS}
\dual{\vu}{\vv} = \e^{\frac{2\pi i}{p}\,\Tr \sym{\vu}{\vv}} \,.
\end{equation}
\end{enumerate}
\end{proposition}
\begin{proof}
(a) Let $W_0$ be any projective representation of $V$ associated with $\Po$. For a fixed direction $D\in\ss$ and line $\lf\in\Af{D}{\Omega}$, we have $\Po(\lf+\vd) = \Po(\lf)$ for all $\vd\in D$. The covariance relation \eqref{eq:defW0} then implies that the operators $\{W_0(\vd) \mid \vd\in D\}$ commute with the rank-$1$ projection $\Po(\lf)$. Therefore, the restriction $\left.W_0\right|_D$ is a projective representation of $D$ which leaves the $1$-dimensional subspace $\hh_0 = \Po(\lf)\hh$ invariant. By Proposition \ref{prop:trivmult} in Appendix \ref{app:projrep}, the representation $\left.W_0\right|_D$ has exact multiplier, hence there exists a function $a_D : D\to\T$ such that $a_D \left.W_0\right|_D$ is an ordinary representation of $D$. In particular, this implies $W_0(\vnull) = \overline{a_D(\vnull)} \id$, hence $a_{D_1}(\vnull) = a_{D_2}(\vnull) \equiv c$ for all $D_1,D_2\in\ss$. Setting $a(\vv) = a_{\F\vv}(\vv)$ for all $\vv\in V\setminus\{\vnull\}$ and $a(\vnull) = c$, item (a) is then satisfied by the projective representation $W=aW_0$.

(b) If $W$ is any projective representation of $V$, by Proposition \ref{prop:bichar} in Appendix \ref{app:projrep} there exists a unique antisymmetric bicharacter $b$ of $V$ such that the equality $W(\vu)W(\vv) = b(\vu,\vv) W(\vv)W(\vu)$ holds for all $\vu,\vv\in V$. Since $b(\vu,\vv)^p = b(p\vu,\vv) = b(\vnull,\vv) = 1$, the bicharacter $b$ takes its values in the set of $p$-roots of unity in $\C$. Hence, there exists a unique function $s:V\times V \to \Z_p$ with $b(\vu,\vv) = \exp(2\pi i \,s(\vu,\vv)/p)$ for all $\vu,\vv\in V$. Since $b$ is antisymmetric, we have $s(\vu,\vv) = - s(\vv,\vu)$. Moreover, the bicharacter property of $b$ and the uniqueness of $s$ easily imply that $s$ is $\Z_p$-bilinear. In particular, fixing a linear basis $\{\ve_1,\ve_2\}$ of $V$ over $\F$, by \cite[Theorem VI.5.2]{LanAlg} for all $i,j=1,2$ there exists a unique element $\sigma_{ij}\in\F$ such that $s(\alpha\ve_i,\ve_j) = \Tr{(\alpha\sigma_{ij})}$ for all $\alpha\in\F$.\\
Now, suppose $W$ is associated with the $V$-covariant quadrature system $\Po$. Then, $W$ is uniquely determined up to a phase function, and the commutation relation \eqref{eq:comm1} does not depend on such a function. Hence by item (a) we can assume that the restrictions $\left.W\right|_D$ are ordinary representations of $D$ for all $D\in\ss$. If $\vv\in V\setminus\{\vnull\}$ and $\alpha\in\F$, taking $D=\F\vv$ this implies that $W(\alpha\vv)W(\vv) = W(\vv)W(\alpha\vv)$, hence $b(\alpha\vv,\vv) = 1$, or, equivalently, $s(\alpha\vv,\vv) = 0$. As a consequence, $s(\alpha\vu,\vv) = s(\vu,\alpha\vv)$ for all $\vu,\vv\in V$ and $\alpha\in\F$, since
\begin{align*}
0 & = s(\alpha(\vu+\vv),\vu+\vv) = s(\alpha\vu,\vu) + s(\alpha\vu,\vv) + s(\alpha\vv,\vu) + s(\alpha\vv,\vv) \\
& = s(\alpha\vu,\vv) - s(\vu,\alpha\vv)
\end{align*}
by $\Z_p$-bilinearity and antisymmetry of $s$.
Introducing the $\F$-bilinear map $S:V\times V\to\F$ defined by $S(\ve_i,\ve_j) = \sigma_{ij}$, we thus see that $s(\vu,\vv) = \Tr{S(\vu,\vv)}$ for all $\vu,\vv\in V$. Indeed,
\begin{align*}
s\Big(\sum_i \alpha_i\ve_i,\sum_j \beta_j\ve_j\Big) & = \sum_{i,j}s(\alpha_i\ve_i,\beta_j\ve_j) = \sum_{i,j}s(\alpha_i\beta_j\ve_i,\ve_j) \\
& = \sum_{i,j}\Tr{(\alpha_i\beta_j\sigma_{ij})} = \Tr{S\Big(\sum_i \alpha_i\ve_i,\sum_j \beta_j\ve_j\Big)}
\end{align*}
for all $\alpha_i,\beta_j\in\F$. The map $S$ is unique by uniqueness of the $\sigma_{ij}$'s and its bilinearity. It remains to show that $S$ is a symplectic form. Since $\Tr{(\alpha S(\vv,\vv))} = s(\alpha\vv,\vv) = 0$ for all $\alpha\in\F$, \cite[Theorem VI.5.2]{LanAlg} yields $S(\vv,\vv) = 0$. To show that $S$ is nonzero, assume by contradiction that $S=0$. Then $W$ is an ordinary representation of the abelian group $V$. If $\vv\neq\vnull$ and $\lf\in\Af{\F\vv}{\Omega}$, we know that the rank-$1$ projection $\Po(\lf)$ commutes with $W(\vv)$, hence $W(\vv)\Po(\lf) = k\Po(\lf)$ for some phase $k\in\T$. It follows that for all $\vu\in V$
\begin{align*}
W(\vv)\Po(\lf+\vu) & = W(\vv)W(\vu)\Po(\lf)W(\vu)^* = W(\vu)W(\vv)\Po(\lf)W(\vu)^* \\
& = k W(\vu)\Po(\lf)W(\vu)^*= k\Po(\lf+\vu) \,.
\end{align*}
If $D\in\ss$ is such that $D\neq \F\vv$, this implies
$$
W(\vv) = W(\vv)\sum_{\vd\in D} \Po(\lf+\vd) = k\id\,.
$$
But this is a contradiction, because if $\mf\in\Af{D}{\Omega}$ the projections $\Po(\mf)$ and $\Po(\mf+\vv) = W(\vv)\Po(\mf)W(\vv)^*$ have orthogonal ranges.
\end{proof}

The symplectic form $S$ uniquely determined by the $V$-covariant quadrature system $\Po$ as in equations \eqref{eq:comm1} and \eqref{eq:defbS} is the symplectic form {\em induced} by $\Po$ on the affine space $(\Omega,V)$. On the other hand, if $(\Omega,V)$ is already a phase-space and its symplectic form $S$ coincides with the one induced by $\Po$, we say that $\Po$ is a $V$-covariant quadrature system {\em for} the phase-space $(\Omega,V,S)$. In both cases, we write $\Po\in\qq_V(\Omega,V,S)$ to highlight the phase-space structure we are dealing with.

Note that, if $\Po\sim\Po'$ and $U$ is any unitary operator intertwining $\Po$ with $\Po'$, then the unitary operators $W'(\vv) = UW(\vv)U^*$ form a projective representation $W'$ of $V$ which is associated with the $V$-covariant quadratures $\Po'$. Since $W'$ has the same commutation relation of $W$, we see that $\Po$ and $\Po'$ induce the same symplectic form on $(\Omega,V)$. However, the converse of this fact remarkably does not hold: indeed, we will see in Section \ref{sec:exis} that there are many inequivalent $V$-covariant quadratures for any fixed phase-space $(\Omega,V,S)$.

\begin{remark}
In \cite{GiHoWo04}, when $W$ is the particular representation defined by \cite[Equation (29)]{GiHoWo04}, a quadrature system for $(\Omega,V)$ satisfying \eqref{eq:defW0} is called a {\em quantum net}. However, we stress that in the present approach no a priori choice is made for $W$, but we rather let it arise from the $V$-covariant quadrature system itself. Actually, fixing the representation $W$ as in \cite{GiHoWo04} is restrictive to some extent, as it does not take into account the possibility that two $V$-covariant quadratures can induce different symplectic forms on $(\Omega,V)$. This affects the partition of the set $\qq_V(\Omega,V)$ into equivalence classes, as it will become clear at the end of Section \ref{sec:exis}.
\end{remark}

By \eqref{eq:comm1} and \eqref{eq:defbS}, the representation $W$ is a particular instance of a {\em Weyl system} \cite{Var95,Zel91,BSZ92,DiHuVa99}, first introduced by Schwinger \cite{Schwinger60} and also known with a wide variety of names in the physics and signal analysis literature: {\em finite Heisenberg group} \cite{BaIt86}, {\em generalized Pauli group} \cite{BaBoRoVa02,ApBeCh08,PaZa88,DuEnBe10}, {\em nice error bases} \cite{Ho05,Kn96,KlRo04,AsChWo07}, {\em translation operators} \cite{GiHoWo04} or {\em displacement operators} \cite{Vourdas04,Vourdas97,Vo07,Vourdas08JFAA}, to cite only the most common ones. It is the finite dimensional analogue of the Schr\" odinger representation of the real Heisenberg group \cite{HAPS89}.

In the present case, Proposition \ref{prop:WHquad} motivates the following refinement of the usual definition of Weyl systems.

\begin{definition}\label{def:WHrep}
Let $(V,S)$ be a $2$-dimensional $\F$-linear symplectic space. A {\em Weyl system} for $(V,S)$ is a projective representation $W$ of $V$ such that
\begin{enumerate}[(i)]
\item for any $D\in\ss$, the restriction $\left. W\right|_D$ of $W$ to $D$ is an ordinary representation of the additive abelian group $D$;
\item the following commutation relation holds:
\begin{equation}\label{eq:commrel}
W(\vu)W(\vv) = \dual{\vu}{\vv} W(\vv)W(\vu) \qquad \forall \vu,\vv\in V
\end{equation}
where $b_S:V\times V\to\T$ is the antisymmetric bicharacter of $V$ given by \eqref{eq:defbS}.
\end{enumerate}
\end{definition}
Note that, if $W$ is any Weyl system, then $W(\vnull) = \id$, and, for all $\vv\in V$, $W(\vv)^\ast =W(-\vv)$.

By Proposition \ref{prop:WHquad}, we can always assume that the projective representation $W$ associated with a $V$-covariant quadrature system $\Po$ is a Weyl system. We call it a Weyl system {\em associated} with $\Po$. However, even restricting to Weyl systems does not remove all the arbitrariness in the choice of the projective representation of $V$ associated with $\Po$. Indeed, suppose $W$ is a Weyl system satisfying \eqref{eq:defW0}, and for all $D\in\ss$ let $\chi_D$ be some character of $D$. Define $W'(\vu) = \chi_{\F\vu}(\vu) W(\vu)$ for all $\vu\neq\vnull$ and $W'(\vnull) = \id$. Then $W$ and $W'$ are two different Weyl systems that are both associated with $\Po$.

In order to remove any ambiguity and to make the choice of $W$ canonical, we need to introduce the next definition.

\begin{definition}
Suppose $W$ is a Weyl system associated with the $V$-covariant quadrature system $\Po$, and let $o\in\Omega$ be any point. Then $W$ is {\em centered} at $o$ if $W(\vd) \Po(o+D) = \Po(o+D)$ for all $D\in\ss$ and $\vd\in D$.
\end{definition}

The following is the uniqueness result we were looking for.

\begin{proposition}\label{prop:uniqW}
If $\Po\in\qq_V(\Omega,V)$ and $o\in\Omega$ is any point, there exists a unique Weyl system $W_o$ associated with $\Po$ and centered at $o$. If $o'\in\Omega$ is another point, then $W_{o'}(\vv) = W_o(\vu_{o,o'}) W_o(\vv) W_o(\vu_{o,o'})^*$ for all $\vv\in V$.
\end{proposition}
\begin{proof}
Existence: Suppose $W$ is a Weyl system associated with $\Po$. For any $D\in\ss$, the restriction $\left.W\right|_D$ is an ordinary representation of $D$ that commutes with the $1$-dimensional projection $\Po(o+D)$, and therefore we have $W(\vd)\Po(o+D) = \chi_D(\vd)\Po(o+D)$ for some character $\chi_D$ of $D$. Setting $W_o(\vv) = \overline{\chi_{\F\vv}(\vv)} W(\vv)$ for all $\vv\neq\vnull$ and $W_o(\vnull) = \id$, the Weyl system $W_o$ is still associated with $\Po$, and it is centered at $o$.

Uniqueness: If the Weyl systems $W_1$ and $W_2$ are both associated with $\Po$ and centered at $o$, then $W_2 = aW_1$ for some phase function $a:V\to\T$ with $a(\vnull)=1$ by Proposition \ref{prop:propW}. Moreover,
$$
a(\vv)\Po(o+\F\vv) = a(\vv)W_1(\vv)\Po(o+\F\vv) = W_2(\vv)\Po(o+\F\vv) = \Po(o+\F\vv)
$$
for all $\vv\neq\vnull$. Hence, $a=1$, and so $W_1 = W_2$.\\
If $o'\neq o$, then setting $W'(\vv) = W_o(\vu_{o,o'}) W_o(\vv) W_o(\vu_{o,o'})^*$ we have $W'(\vv) = \dual{\vu_{o,o'}}{\vv} W_o(\vv)$, therefore $W'$ is another Weyl system associated with $\Po$. Since, for all $D\in\ss$ and $\vd\in D$,
\begin{align*}
W'(\vd)\Po(o'+D) & = W_o(\vu_{o,o'}) W_o(\vd) W_o(\vu_{o,o'})^*\Po(o+\vu_{o,o'}+D) \\
& = W_o(\vu_{o,o'}) W_o(\vd) \Po(o+D) W_o(\vu_{o,o'})^* \\
& = W_o(\vu_{o,o'}) \Po(o+D) W_o(\vu_{o,o'})^*\\
& = \Po(o'+D) \,,
\end{align*}
$W'$ is centered at $o'$, hence $W'=W_{o'}$ by the uniqueness statement.
\end{proof}

The relation between quadratures and Weyl systems is very well known, both in the case $\F=\R$ \cite{AlDVTo09,LaPe10,KiSc13,CaHeScTo14} and when $\F$ is a finite field as in the present paper \cite{GiHoWo04,DuEnBe10,SuTo07,ShVo11bis,ApDaFu14}. In the latter case, the use of Weyl systems to construct quadrature systems essentially goes back to \cite{BaBoRoVa02}. In the next two sections, we will refine this construction and use it to determine all the equivalence classes of $V$-covariant quadrature systems.

\section{Equivalence of $V$-covariant quadrature systems: from $V$-covariant quadratures to Weyl multipliers}\label{sec:equiv}

For any choice of the symplectic form $S$ on $V$, $V$-covariant quadrature systems for the phase-space $(\Omega,V,S)$ actually exist and they are grouped into a finite collection of equivalence classes. This is the main content of the present and the next sections, and, as we will shortly see, the claim is a consequence of a detailed analysis of the associated Weyl systems and their multipliers. (A quick review on multipliers and their main properties used in the paper can be found in Appendix \ref{app:projrep}).

In this section, we will concentrate on the equivalence problem, while the proof of the existence will be deferred to the next one. More precisely, here we will prove the following two main facts:
\begin{enumerate}[(a)]
\item Weyl systems associated with $V$-covariant quadrature systems are irreducible;
\item two $V$-covariant quadratures are equivalent if and only if their associated centered Weyl systems are such.
\end{enumerate}
(Irreducibility and equivalence of Weyl systems is understood in the usual sense of projective representations, see again Appendix \ref{app:projrep}).
Combining these two facts, the problem of classifying all the equivalence classes of $V$-covariant quadratures descends to the same but easier task for irreducible Weyl systems. Indeed, Stone-von Neumann theorem then applies, which states that two irreducible Weyl systems are equivalent if and only if their multipliers are equal. So, we will end up with a very simple characterization: two $V$-covariant quadrature systems are equivalent if and only if their associated centered Weyl systems have the same multiplier. This turns the classification problem for $V$-covariant quadratures into the analogous problem for a special class of multipliers, that is, the class of the Weyl multipliers which we define at the end of the section.

It will be shown in a moment that the relation between $V$-covariant quadratures and associated Weyl systems is established by Fourier transform along the directions of $\Omega$. But before doing this, we need the following precise analysis of the group $\hat{V}$ of characters of $V$.

\begin{proposition}
For any symplectic form $S$ on $V$, the map $\vv\mapsto\dual{\cdot}{\vv}$ is a group isomorphism of $V$ onto its character group $\hat{V}$. It maps each subgroup $D\in\ss$ onto its annihilator subgroup $D^\perp := \{\chi\in\hat{V}\mid \left.\chi\right|_D = 1\}$, and thus establishes a group isomorphism of the quotient group $V/D$ with the character group $\hat{D}$ of $D$.
\end{proposition}
\begin{proof}
The map $\vv\mapsto\dual{\cdot}{\vv}$ clearly is a group homomorphism of $V$ into $\hat{V}$. Since $|V| = |\hat{V}|$, in order to prove that it is an isomorphism it suffices to show its injectivity. If $\dual{\cdot}{\vv} = 1$, then $\Tr{\sym{\alpha\vw}{\vv}} = \Tr{(\alpha\sym{\vw}{\vv})} = 0$ for all $\vw\in V$ and $\alpha\in\F$, which implies $\vv = \vnull$ by nondegeneracy of the symplectic form $\sym{\cdot}{\cdot}$ and of the $\Z_p$-bilinear map $\F\times\F \ni (\alpha,\beta) \mapsto \Tr{(\alpha\beta)} \in \Z_p$ \cite[Theorem VI.5.2]{LanAlg}.\\
To prove the second claim, note that the character $\dual{\cdot}{\vd}\in D^\perp$ for all $\vd\in D$. On the other hand, by the canonical isomorphism $\hat{D}\simeq \hat{V}/D^\perp$ \cite[Corollary I.9.3]{LanAlg}, we have $|\hat{D}| = |\hat{V}|/|D^\perp|$, and then, since $|\hat{D}| = |D| = |\F|$ and $|\hat{V}| = |\F|^2$, it follows that $|D| = |D^\perp|$. Hence, the map $\vd\mapsto\dual{\cdot}{\vd}$ from $D$ to $D^\perp$ is onto. Finally, the isomorphism statement is a consequence of the just proved identifications $V\simeq \hat{V}$, $D\simeq D^\perp$ and the isomorphism $\hat{D}\simeq \hat{V}/D^\perp$.
\end{proof}

With the identification $\hat{D} = V/D$, the {\em orthogonality relations} for characters of $D$ \cite[Theorem XVIII.5.2]{LanAlg} give the formula
\begin{equation}\label{eq:ortorel}
\sum_{\vd\in D} \dual{\vv-\vu}{\vd} = |\F| \delta_{\vu+D,\vv+D} \qquad \forall \vu,\vv\in V \,.
\end{equation}
Using it, we obtain a direct link between $V$-covariant quadratures and associated Weyl systems.

\begin{proposition}\label{prop:Fourier_quad}
Suppose $\Po$ is a $V$-covariant quadrature system for the phase-space $(\Omega,V,S)$ acting on the Hilbert space $\hh$, and let $W_o$ be its associated Weyl system centered at $o$. Then, for all $\vu\neq\vnull$,
\begin{align}\label{eq:defW}
W_o(\vu) = \sum_{\vv+\F\vu\in V/\F\vu} \dual{\vu}{\vv} \Po(o+\vv+\F\vu)
\end{align}
and, for all $D\in\ss$ and $\vv\in V$,
\begin{align}\label{eq:defP}
\Po(o+\vv+D) = \frac{1}{|\F|} \sum_{\vd\in D} \dual{\vv}{\vd} W_o(\vd) \,.
\end{align}
\end{proposition}
\begin{proof}
Recalling that the quotient group $V/\F\vu$ acts freely and transitively on the set of parallel lines $\Af{\F\vu}{\Omega}$, we have
\begin{align*}
W_o(\vu) & = W_o(\vu) \sum_{\lf\in\Af{\F\vu}{\Omega}} \Po(\lf) = W_o(\vu) \sum_{\vv+\F\vu\in V/\F\vu} \Po(o+\vv+\F\vu) \\
& = W_o(\vu) \sum_{\vv+\F\vu\in V/\F\vu} W_o(\vv)\Po(o+\F\vu)W_o(\vv)^* \\
& = \sum_{\vv+\F\vu\in V/\F\vu} \dual{\vu}{\vv} W_o(\vv)W_o(\vu)\Po(o+\F\vu)W_o(\vv)^* \\
& = \sum_{\vv+\F\vu\in V/\F\vu} \dual{\vu}{\vv} W_o(\vv)\Po(o+\F\vu)W_o(\vv)^* \\
& = \sum_{\vv+\F\vu\in V/\F\vu} \dual{\vu}{\vv} \Po(o+\vv+\F\vu) \,.
\end{align*}
Using the orthogonality relations \eqref{eq:ortorel}, for all $\vw\in V$ we then have
\begin{align*}
\sum_{\vd\in D} \dual{\vw}{\vd} W_o(\vd) & = \sum_{\vd\in D} \dual{\vw}{\vd} \sum_{\vv+D\in V/D} \dual{\vd}{\vv} \Po(o+\vv+D) \\
& = \sum_{\vv+D\in V/D} \sum_{\vd\in D} \dual{\vw-\vv}{\vd} \Po(o+\vv+D) \\
& = |\F|\Po(o+\vw+D) \,,
\end{align*}
which is \eqref{eq:defP}.
\end{proof}

\begin{corollary}\label{cor:irred}
Any Weyl system associated with a $V$-covariant quadrature system is irreducible.
\end{corollary}
\begin{proof}
The operators $\{\Po(o+\vv+D)\mid \vv\in V,\ D\in\ss\}$ span the linear space $\lh$ by Proposition \ref{prop:reconI}, hence so do the operators $\{W_o(\vv)\mid \vv\in V\}$ by \eqref{eq:defP}. In particular, the subalgebra $\aa$ of $\lh$  generated by the latter operators coincides with $\lh$, hence $W_o$ is irreducible. Since all the Weyl systems associated with $\Po$ only differ by phase functions, the same holds for any of them.
\end{proof}

\begin{corollary}\label{cor:equiv}
Suppose $\Po_1,\Po_2\in\qq_V(\Omega,V)$, and let $W_1$ and $W_2$ be Weyl systems associated with $\Po_1$ and $\Po_2$, respectively. Assume that $W_i$ is centered at $o_i$, possibly with $o_1\neq o_2$.  Then a unitary operator $U$ intertwines $\Po_1$ with $\Po_2$ if and only if $W_2(\vv) = UW_1(\vu_{o_1,o_2})W_1(\vv)W_1(\vu_{o_1,o_2})^*U^*$ for all $\vv\in V$. In particular, $\Po_1\sim\Po_2$ if and only if their centered Weyl systems $W_1$ and $W_2$ are equivalent.
\end{corollary}
\begin{proof}
By Proposition \ref{prop:uniqW}, the Weyl system $W_1'$ associated with $\Po_1$ and centered at $o_2$ is $W_1'(\vv) = W_1(\vu_{o_1,o_2})W_1(\vv)W_1(\vu_{o_1,o_2})^*$. Since $W_1'$ and $W_2$ are both centered at $o_2$, by formulas \eqref{eq:defW} and \eqref{eq:defP} a unitary operator $U$ intertwines $\Po_1$ with $\Po_2$ if and only if $W_2(\vv) = UW_1'(\vv)U^*$. The claim then follows.
\end{proof}

Corollaries \ref{cor:irred} and \ref{cor:equiv} turn the classification problem for $V$-covariant quadrature systems into the analogous task for their associated irreducible Weyl systems, as anticipated at the beginning of this section. On the other hand, the classification of {\em all} the irreducible Weyl systems is provided by the following variant of Stone-von Neumann theorem (cf.~\cite[Theorem 1]{Mackey49SvN}).

\begin{proposition}\label{prop:SvNrep}
Suppose $W$ is an irreducible Weyl system which acts on the Hilbert space $\hh$. Then $\dim\hh = |\F|$, and the set $\{|\F|^{-1/2} W(\vv)\mid \vv\in V\}$ is an orthonormal basis of the linear space $\lh$ endowed with the Hilbert-Schmidt inner product. If $W'$ is another irreducible Weyl system, then $W$ and $W'$ are equivalent if and only if they have the same multiplier.
\end{proposition}
\begin{proof}
Suppose $W$ is a Weyl system for the symplectic space $(V,S)$. For all $\vu\in V$, let $\Phi(\vu):\lh\to\lh$ be the linear map $[\Phi(\vu)](A) = W(\vu)AW(\vu)^*$. Then $\Phi(\vu)$ is a unitary operator on $\lh$ endowed with the Hilbert-Schmidt inner product, and the map $\Phi : V\to\uelle{\lh}$ is an ordinary representation of $V$ in $\lh$. Note that $[\Phi(\vu)](W(\vv)) = \dual{\vu}{\vv} W(\vv)$, that is, $\Phi$ acts as the character $\dual{\cdot}{\vv}$ on the $1$-dimensional subspace spanned by the operator $W(\vv)$. Since the characters $\dual{\cdot}{\vv_1}$ and $\dual{\cdot}{\vv_2}$ are different if $\vv_1\neq \vv_2$, the set $\{W(\vv)\mid \vv\in V\}$ constitutes an orthogonal sequence in $\lh$ by a standard argument.\\
Let $\aa$ be the linear subalgebra of $\lh$ generated by the operators $\{W(\vv)\mid \vv\in V\}$. As $W(\vv_1)W(\vv_2) = \overline{m(\vv_1,\vv_2)} W(\vv_1+\vv_2)$, where $m$ is the multiplier of $W$, the algebra $\aa$ actually coincides with the linear span of $\{W(\vv)\mid \vv\in V\}$. If $W$ is irreducible, we have $\aa = \lh$ \cite[Corollary 1.17]{Isaacs06}, hence the set $\{W(\vv)\mid \vv\in V\}$ is an orthogonal basis of $\lh$. In particular, $\dim\lh = |V| = |\F|^2$, which implies that $\dim\hh = |\F|$. The normalization constant $|\F|^{-1/2}$ then comes from the fact that $\ip{W(\vu)}{W(\vu)}_{HS} = \tr{\id} = |\F|$.\\
Finally, suppose $W$ and $W'$ are two irreducible Weyl systems acting on the Hilbert spaces $\hh$ and $\hh'$. If they are equivalent, then they clearly have the same associated multiplier. Conversely, if the multipliers of $W$ and $W'$ coincide, the map $\Psi:\lh\to\elle{\hh'}$ defined by $\Psi(W(\vu)) = W'(\vu)$ for all $\vu\in V$ is an isomorphism of $C^*$-algebras. Hence there is a unitary operator $U:\hh\to\hh'$ such that $\Psi(A) = UAU^*$ for all $A\in\lh$ \cite[Proposition 1.5]{Gu67}, that is, $W$ and $W'$ are equivalent.
\end{proof}

Corollaries \ref{cor:irred}, \ref{cor:equiv} and Proposition \ref{prop:SvNrep} suggest to characterize the equivalence of $V$-covariant quadratures through the multipliers of their associated Weyl systems. Indeed, let us define the {\em associated multiplier} of a $V$-covariant quadrature system $\Po$ to be the multiplier of the Weyl system associated with $\Po$ and centered at an arbitrary point $o\in\Omega$. This definition is consistent, since by Corollary \ref{cor:equiv} such a multiplier is unaffected by the choice of $o$, and only depends on the equivalence class of $\Po$. We then obtain the following characterization.

\begin{proposition}\label{prop:quad-->mult}
Two $V$-covariant quadrature systems are equivalent if and only if they have the same associated multiplier.
\end{proposition}
\begin{proof}
The proof is immediate by combining Corollaries \ref{cor:irred}, \ref{cor:equiv} and Proposition \ref{prop:SvNrep}
\end{proof}

Therefore, the equivalence classes of $V$-covariant quadrature systems are unambiguously labeled by the respective associated multipliers. This suggests to single out the essential properties of such multipliers in the next definition.

\begin{definition}\label{def:WHmul}
A {\em Weyl multiplier} for the symplectic space $(V,S)$ is any multiplier of the additive group $V$ satisfying the following two conditions:
\begin{enumerate}[(i)]
\item for any $D\in\ss$, $m(\vd_1,\vd_2) = 1$ for all $\vd_1,\vd_2\in D$;
\item $\overline{m(\vu,\vv)} m(\vv,\vu) = \dual{\vu}{\vv}$ for all $\vu,\vv\in V$.
\end{enumerate}
\end{definition}

We will denote by $\mm(V,S)$ the set of Weyl multipliers for $(V,S)$. Observe that any $m\in\mm(V,S)$ satisfies $m(\vu,\vnull) = m(\vnull,\vu) = m(\vu,-\vu) = 1$ for all $\vu\in V$. However, a Weyl multiplier {\em is not} exact.

It is easily checked that the Weyl systems for the symplectic space $(V,S)$ are exactly the projective representations of $V$ whose multipliers are Weyl multipliers for $(V,S)$. In particular, the multiplier assciated with any quadrature system in the set $\qq_V(\Omega,V,S)$ is a Weyl multiplier in $\mm(V,S)$. This fact motivates a deeper analysis of Weyl multipliers, which will be the topic of the next section.

\section{Existence of $V$-covariant quadrature systems: from Weyl multipliers to $V$-covariant quadratures}\label{sec:exis}

By Proposition \ref{prop:quad-->mult} two equivalence classes of $V$-covariant quadratures are equivalent if and only if they have the same associated Weyl multiplier. Now, Theorem \ref{teo:quad=mult} below will prove that for any multiplier $m\in\mm(V,S)$ there exists a quadrature system in $\qq_V(\Omega,V,S)$ having $m$ as its associated multiplier. Therefore, the existence problem for $V$-covariant quadratures actually turns into the corresponding problem for Weyl multipliers. This explains the relevance of Weyl multipliers in the description of $V$-covariant quadratures, and leads us to completely characterize the set $\mm(V,S)$ in the next proposition.

\begin{proposition}\label{prop:SvNmult}
The set $\mm(V,S)$ is nonempty and finite, with cardinality $|\mm(V,S)| = |\F|^{|\F|-1}$. Moreover, any two multipliers $m_1,m_2\in\mm(V,S)$ are equivalent, and, if $a:V\to\T$ is a function intertwining $m_1$ with $m_2$, then for any $D\in\ss$ the restriction $\left. a\right|_D$ is a character of $D$.
\end{proposition}
\begin{proof}
Existence: Choose a symplectic basis $\{\ve_1,\ve_2\}$ of $V$, and define the following map $m_0 : V\times V \to \T$
$$
m_0(\alpha_1 \ve_1 + \alpha_2 \ve_2 \, , \, \beta_1 \ve_1 + \beta_2 \ve_2) = \e^{\frac{2\pi i}{p} \,\Tr{(\beta_1\alpha_2)}} \, .
$$
It is easy to check that $m_0$ is a multiplier of $V$ which satisfies the condition $\overline{m_0(\vu,\vv)} m_0(\vv,\vu) = \dual{\vu}{\vv}$ for all $\vu,\vv\in V$. Moreover, for $\alpha_1,\alpha_2,\gamma\in\F$, we have that $m_0(\gamma(\alpha_1 \ve_1 + \alpha_2 \ve_2) \, , \, \alpha_1 \ve_1 + \alpha_2 \ve_2) = m_0(\alpha_1 \ve_1 + \alpha_2 \ve_2 \, , \, \gamma(\alpha_1 \ve_1 + \alpha_2 \ve_2))$. This means that, for any $D\in\ss$, $m_0(\vd_1,\vd_2) = m_0(\vd_2,\vd_1)$ for all $\vd_1,\vd_2\in D$, hence there exists a function $a_D : D \to \T$ such that $m_0(\vd_1,\vd_2) = \overline{a_D(\vd_1)} \overline{a_D(\vd_2)} a_D(\vd_1 + \vd_2)$ by Proposition \ref{prop:abelrep} in Appendix \ref{app:projrep}. Note that $\overline{a_D(\vnull)} = m_0(\vnull,\vnull) = 1$, hence, setting $a_{\{\vnull\}}(\vnull) = 1$, the map $m : V\times V \to \T$ with
$$
m(\vu , \vv) = a_{\F\vu}(\vu) a_{\F\vv}(\vv) \overline{a_{\F(\vu+\vv)}(\vu+\vv)} m_0(\vu , \vv)
$$
is a Weyl multiplier for $(V,S)$. 

Uniqueness: If $m_1,m_2\in\mm(V,S)$, then the multiplier $\overline{m_1} m_2$ satisfies $(\overline{m_1} m_2)(\vu,\vv) = (\overline{m_1} m_2)(\vv,\vu)$ for all $\vu,\vv\in V$, hence $(\overline{m_1} m_2)(\vu,\vv) = \overline{a(\vu)a(\vv)} a(\vu+\vv)$ for some function $a:V\to\T$ by Proposition \ref{prop:abelrep} in Appendix \ref{app:projrep}. That is, $m_1$ and $m_2$ are equivalent and intertwined by $a$. For $D\in\ss$ and $\vd_1,\vd_2\in D$, we have $m_1(\vd_1,\vd_2) = m_2(\vd_1,\vd_2) = 1$, hence $a(\vd_1+\vd_2) = a(\vd_1)a(\vd_2)$, i.e., $\left. a\right|_D$ is a character of $D$.

Finiteness: Fix an element $m\in\mm(V,S)$. Then, any other $m'\in\mm(V,S)$ is obtained from $m$ by picking a phase function $a:V\to\T$ such that $\left. a\right|_D \in \hat{D}$ for all $D\in\ss$, and letting
$$
m'(\vu , \vv) = \overline{a(\vu)a(\vv)} a(\vu+\vv) m(\vu , \vv) \qquad \forall \vu,\vv\in V\,.
$$
The set of functions $\ff = \{a:V\to\T \mid \left. a\right|_D \in \hat{D}\ \forall D\in\ss\}$ has cardinality $|\ff| = |\hat{D}|^{|\ss|} = |\F|^{|\F|+1}$. Moreover, two multipliers $m'_1,m'_2\in\mm(V,S)$ coincide if and only if the phase functions $a_1$ and $a_2$, intertwining $m$ with $m'_1$ and $m'_2$, respectively, as in the previous formula, only differ up to a character of $V$. That is, $m'_1 = m'_2$ if and only if there exists a character $\chi\in\hat{V}$ such that $a_2(\vv) = \chi(\vv)a_1(\vv)$ for all $\vv\in V$. Therefore, we have $|\mm(V,S)| = |\ff|/|\hat{V}| = |\F|^{|\F|-1}$.
\end{proof}

We now give an explicit example of a Weyl multiplier. It is the finite field analogue of the well known multiplier $m((q_1,p_1),(q_2,p_2)) = \e^{i(q_1 p_2 - q_2 p_1)/2}$ of a Weyl system on $\R^2$ \cite[Theorem 7.38]{GQT85}.

\begin{example}\label{ex:multodd}
If $\F$ has characteristic $p\neq 2$, then $m(\vu,\vv) = \dual{2^{-1}\vv}{\vu}$ is a Weyl multiplier for $(V,S)$. As we will see in the next section, such a multiplier is special, since it is the unique element of $\mm(V,S)$ having the remarkable property of being invariant under the action of the symplectic group of $(V,S)$. In characteristic $p=2$, however, the explicit construction of a Weyl multiplier is more involved (see \ref{app:mult}), and, contrary to the case $p\neq 2$, there exists no distinguished element in $\mm(V,S)$.
\end{example}

The following theorem is the main result in our characterization of $V$-covariant quadrature systems. Indeed, we have established a correspondence $\qq_V(\Omega,V,S)\mapsto\mm(V,S)$ which sends each quadrature system of $\qq_V(\Omega,V,S)$ into its associated multiplier in $\mm(V,S)$. By Proposition \ref{prop:quad-->mult}, this correspondence factors to an injective mapping on the set of equivalence classes of quadratures. The next theorem  proves that such a mapping is onto, and thus establishes the fundamental equivalence between $V$-covariant quadrature systems and Weyl multipliers.

\begin{theorem}\label{teo:quad=mult}
Suppose $(\Omega,V,S)$ is a phase-space. For any element $m\in\mm(V,S)$, there exists a unique equivalence class $\qq^m_V(\Omega,V,S)$ of $V$-covariant quadrature systems for $(\Omega,V,S)$ whose associated multiplier is $m$. If $W$ is an irreducible Weyl system acting on $\hh$ and having multiplier $m$ and $o\in\Omega$ is any point, the map $\Po:\Aff{\Omega}\to\lh$ given by
\begin{equation}\label{eq:defPbis}
\Po(o+\vv+D) = \frac{1}{|\F|} \sum_{\vd\in D} \dual{\vv}{\vd} W(\vd) \qquad \forall\vv\in V,\ D\in\ss
\end{equation}
is a $V$-covariant quadrature system in $\qq^m_V(\Omega,V,S)$ and $W$ is its associated Weyl system centered at $o$.
\end{theorem}
\begin{proof}
The uniqueness of the equivalence class $\qq_V^m(\Omega,V,S)$ follows from Proposition \ref{prop:quad-->mult}.\\
By Proposition \ref{prop:exrep} in Appendix \ref{app:projrep}, there exists an irreducible Weyl system $W$ for the symplectic space $(V,S)$ whose multiplier is $m$. Hence it is enough to show that for such a $W$ formula \eqref{eq:defPbis} defines an element $\Po\in\qq_V(\Omega,V)$, and that $W$ is the Weyl system associated with $\Po$ and centered at $o$.\\
It is easy to check that $\Po(o+\vv+D)^\ast = \Po(o+\vv+D)$. Moreover, since $\left.W\right|_D$ is an ordinary representation,
\begin{align*}
& \Po(o+\vv_1+D)\Po(o+\vv_2+D) = \\
& \qquad \qquad = \frac{1}{|\F|^2} \sum_{\vd_1,\vd_2\in D} \dual{\vv_1}{\vd_1} \dual{\vv_2}{\vd_2} W(\vd_1+\vd_2)\\
& \qquad \qquad = \frac{1}{|\F|^2} \sum_{\vd_2\in D} \dual{\vv_2-\vv_1}{\vd_2} \sum_{\vd'_1\in D} \dual{\vv_1}{\vd'_1} W(\vd'_1)\\
& \qquad \qquad = \delta_{\vv_1+D,\vv_2+D} \Po(o+\vv_1+D) \,,
\end{align*}
in which we made the substitution $\vd'_1 = \vd_1+\vd_2$ and we used the orthogonality relations \eqref{eq:ortorel}.
Therefore, the operators $\{\Po(o+\vv+D)\mid \vv+D\in V/D\}$ are orthogonal projections, and the ranges of $\Po(o+\vv_1+D)$ and $\Po(o+\vv_2+D)$ are orthogonal if $\vv_1+D \neq \vv_2+D$. Since $|V/D| = |\F| = \dim\hh$ by Proposition \ref{prop:SvNrep}, each projection $\Po(o+\vv+D)$ then must have rank $1$, and $\sum_{\vv+D\in V/D} \Po(o+\vv+D) = \id$.\\
For any $\vu\in V$, by the commutation relation \eqref{eq:commrel} we have
\begin{align}\label{eq:covarPo}
\begin{aligned}
W(\vu)\Po(o+\vv+D)W(\vu)^* & = \frac{1}{|\F|} \sum_{\vd\in D} \dual{\vv}{\vd} \dual{\vu}{\vd} W(\vd) \\
& = \Po(o+\vv+\vu+D) \,, 
\end{aligned}
\end{align}
hence $\Po(\lf+\vu) = W(\vu)\Po(\lf)W(\vu)^*$ for all $\lf\in\Aff{\Omega}$ and $\vu\in V$.\\
In order to show that $\Po$ is a $V$-covariant quadrature system, we still need to prove the mutual unbiasedness relation in Definition \ref{def:quad}. If $D_1\neq D_2$ and $\lf_i \in \Af{D_i}{\Omega}$, then, for all $\vd\in D_1$,
\begin{align*}
\tr{\Po(\lf_1)\Po(\lf_2)} & = \tr{\Po(\lf_1-\vd)\Po(\lf_2)} = \tr{W(\vd)^*\Po(\lf_1)W(\vd)\Po(\lf_2)}\\
& = \tr{\Po(\lf_1)W(\vd)\Po(\lf_2)W(\vd)^*} \\
& = \tr{\Po(\lf_1)\Po(\lf_2+\vd)} \,.
\end{align*}
On the other hand,
$$
\sum_{\vd\in D_1} \tr{\Po(\lf_1)\Po(\lf_2+\vd)} = \tr{\Po(\lf_1)\sum_{\mf\in\Af{D_2}{\Omega}} \Po(\mf)} = \tr{\Po(\lf_1)} = 1 \,.
$$
Combining these two facts, we see that $\tr{\Po(\lf_1)\Po(\lf_2)} = 1/|\F|$, which completes our proof that $\Po$ is a $V$-covariant quadrature system.\\
We now show that $W$ is the Weyl system associated with $\Po$ and centered at $o$. We have already seen in \eqref{eq:covarPo} that $W$ is a Weyl system associated with $\Po$. Moreover, for all $D\in\ss$ and $\vd\in D$,
$$
W(\vd)\Po(o+D) = \frac{1}{|\F|} \sum_{\vd'\in D} W(\vd+\vd') = \Po(o+D) \,,
$$
that is, $W$ is centered at $o$.
\end{proof}

By the existence of Weyl multipliers proved in Proposition \ref{prop:SvNmult}, Theorem \ref{teo:quad=mult} thus implies that the set $\qq_V(\Omega,V,S)$ is nonempty for any symplectic form $S$ on $V$, that is, $V$-covariant quadrature systems exist for any phase-space $(\Omega,V,S)$. Moreover, it shows that the set $\qq_V(\Omega,V)$ is partitioned into the disjoint union of the equivalence classes
$$
\qq_V(\Omega,V) = \bigcup_{S\in\Sym} \bigcup_{m\in\mm(V,S)} \qq_V^m(\Omega,V,S)
$$
where $\Sym$ is the collection of all symplectic forms on $V$. Since $|\Sym| = |\F_*| = |\F|-1$ and for all $S\in\Sym$ we have $|\mm(V,S)| = |\F|^{|\F|-1}$ by Proposition \ref{prop:SvNmult}, the previous union involves $(|\F|-1)|\F|^{|\F|-1}$ equivalence classes. In particular, for all $S\in\Sym$ there are at least two distinct equivalence classes in $\qq_V(\Omega,V,S)$, that is, inequivalent $V$-covariant quadratures for the same phase-space $(\Omega,V,S)$ actually exist, as we anticipated in Section \ref{sec:Weylsyst}.

The number $(|\F|-1)|\F|^{|\F|-1}$ of equivalence classes in the set $\qq_V(\Omega,V)$ should be compared with the analogous number $|\F|^{|\F|-1}$ of inequivalent quantum nets found in \cite[Section VI]{GiHoWo04}, where, however, the authors did not consider the possibility that the symplectic form $S$ associated with different $V$-covariant quadratures may vary within the set $\Sym$.

As an important consequence of the uniqueness statement for Weyl multipliers contained in Proposition \ref{prop:SvNmult}, the set $\qq_V(\Omega,V,S)$ can be characterized actually using {\em a single} quadrature system $\Po\in\qq_V(\Omega,V,S)$. Indeed, by the next proposition one can pass from $\Po$ to any other quadratures $\Po'\in\qq_V(\Omega,V,S)$ simply by relabeling the lines of $\Po$.

\begin{proposition}\label{prop:permutation}
Let $\Po_1,\Po_2\in\qq_V(\Omega,V,S)$, where each $\Po_i$ acts on the Hilbert space $\hh_i$. Then there exist a unitary operator $U:\hh_1\to\hh_2$ and, for all directions $D\in\ss$, a vector $\vv_D\in V$ such that
\begin{equation}\label{eq:permutation}
\Po_2(\lf) = U\Po_1(\lf+\vv_D)U^* \qquad \forall \lf\in\Af{D}{\Omega},\,D\in\ss\,.
\end{equation}
\end{proposition}
\begin{proof}
Suppose $o\in\Omega$ is a fixed point. Let $W_i$ be the Weyl system associated with the $V$-covariant quadrature system $\Po_i$ and centered at $o$, and let $m_i$ be its Weyl multiplier. By Proposition \ref{prop:SvNmult}, for all $D\in\ss$ there is a character $\chi_D \in \hat{D}$ such that
$$
m_2(\vu,\vv) = \overline{\chi_{\F\vu}(\vu) \chi_{\F\vv}(\vv)} \chi_{\F(\vu+\vv)}(\vu+\vv) m_1(\vu,\vv) \qquad\forall\vu,\vv\in V\setminus\{\vnull\} \,.
$$
Therefore, if we define the projective representation $W'_1$ of $V$ in $\hh_1$, with
$$
W'_1 (\vu) = \chi_{\F\vu}(\vu) W_1 (\vu) \qquad \forall \vu\in V\setminus\{\vnull\} \qquad \text{and} \qquad W'(\vnull) = \id \,,
$$
the multiplier of $W'_1$ is $m_2$. By the identification $\hat{D} = V/D$, we have $\chi_D = \overline{\dual{\cdot}{\vv_D}}$ for some vector $\vv_D + D\in V/D$. Let $\Po'_1$ be the $V$-covariant quadrature system with
$$
\Po'_1(x+D) = \Po_1(x+\vv_D+D) \qquad \forall x+D\in\Aff{\Omega}\,.
$$
It is easy to check that $W'_1$ is a Weyl system associated with $\Po'_1$. Moreover, for all $D\in\ss$ and $\vd\in D$,
\begin{align*}
& W'_1 (\vd) \Po'_1 (o+D) = \overline{\dual{\vd}{\vv_D}} W_1(\vd) W_1(\vv_D) \Po_1 (o+D) W_1(\vv_D)^* \\
& \quad = W_1(\vv_D) W_1(\vd) \Po_1 (o+D) W_1(\vv_D)^* = W_1(\vv_D) \Po_1 (o+D) W_1(\vv_D)^* \\
& \quad = \Po'_1 (o+D) \,,
\end{align*}
hence $W'_1$ is the Weyl system associated with $\Po'_1$ and centered at $o$. By Proposition \ref{prop:quad-->mult}, $\Po_1'$ and $\Po_2$ are equivalent, hence \eqref{eq:permutation} follows.
\end{proof}

Proposition \ref{prop:permutation} states that any two quadrature systems $\Po_1$ and $\Po_2\in\qq_V(\Omega,V,S)$ only differ by cyclic permutations of the parallel lines in the sets $\Af{D}{\Omega}$, each permutation depending on the common direction $D$ of the lines. In particular, it implies that the ranges of all $V$-covariant quadrature systems for the phase-space $(\Omega,V,S)$ are unitarily conjugated: that is, if $\Po_1,\Po_2\in\qq_V(\Omega,V,S)$, there exists a unitary operator $U$ such that $\ran\Po_2 = U(\ran\Po_1) U^*$, where $\ran\Po = \{\Po(\lf)\mid\lf\in\Aff{\Omega}\} \subset \lh$. Actually, we will see in Theorem \ref{teo:permutation2} of the next section that the conjugacy of the ranges is a general property of $V$-covariant quadrature systems, and it is not only restricted to systems inducing the same symplectic form on $(\Omega,V)$.

\section{The action of the symplectic group on $V$-covariant quadratures}\label{sec:sympl}

In this section, we enlarge our covariance group and study quadrature systems that are covariant with respect to subgroups $G\subseteq {\rm GL}(V)\rtimes V$ properly containing the translation group $V$. By Proposition \ref{prop:propW}, this will lead us to consider projective representations of the semidirect product $G_0\rtimes V$, where $G_0 = G\cap {\rm GL}(V)$, which are extensions of Weyl systems on $V$. However, it will soon become clear that not all covariance subgroups are allowed. Indeed, this is a consequence of the next easy but very useful observation.

\begin{proposition}\label{prop:azGL}
Let $m\in\mm(V,S)$ and $\Po\in\qq^m_V(\Omega,V,S)$. Moreover, let $W$ be the Weyl system associated with $\Po$ and centered at the unique point $o\in\Omega$ such that ${\rm GL}(V)\cdot o = \{o\}$. Given $A\in{\rm GL}(V)$, define the projective representation $W_A$ of $V$ with
\begin{equation}\label{eq:defWA}
W_A(\vv) = W(A\vv) \qquad \forall \vv\in V \,.
\end{equation}
Then $W_A$ is the Weyl system associated with the $V$-covariant quadratures $\Po_A$ and centered at the point $o$. Furthermore, we have $\Po_A\in\qq^{m_A}_V(\Omega,V,S_A)$, where $S_A$ is the symplectic form $S_A(\cdot,\cdot) = S(A\cdot,A\cdot)$ and $m_A\in\mm(V,S_A)$ is the multiplier $m_A(\cdot,\cdot) = m(A\cdot,A\cdot)$.
\end{proposition}
\begin{proof}
By definitions,
\begin{align*}
W_A(\vv)\Po_A(\lf)W_A(\vv)^* & = W(A\vv)\Po(A\cdot \lf)W(A\vv)^* = \Po(A\cdot \lf+A\vv) \\
& = \Po_A(\lf+\vv)
\end{align*}
for all $\lf\in\Aff{\Omega}$ and $\vv\in V$, and
$$
W_A(\vd)\Po_A(o+D) = W(A\vd)\Po(o+AD) = \Po(o+AD)
$$
for all $D\in\ss$ and $\vd\in D$ since $A\vd\in AD$. This proves the first claim. For the second, we have
\begin{align*}
W_A(\vu)W_A(\vv) &= \e^{\frac{2\pi i}{p}\,\Tr{S(A\vu,A\vv)}}W_A(\vv)W_A(\vu) \\
& = \overline{m(A\vu,A\vv)} W_A(\vu+\vv)\qquad \forall \vu,\vv\in V \,,
\end{align*}
that is, the quadrature system $\Po_A$ induces the symplectic form $S_A$, and $m_A$ is its associated Weyl multiplier.
\end{proof}

Since all symplectic forms only differ by a nonzero scalar, we have $S_A = \lam(A) S$ for some  $\lam(A)\in\F_*$. To determine $\lam(A)$, write $A\ve_i = \alpha_{1i}\ve_1 + \alpha_{2i}\ve_2$ with respect to some symplectic basis $\{\ve_1,\ve_2\}$ of $(V,S)$. Then
$$
\lam(A) = \lam(A)S(\ve_1,\ve_2) = S(A\ve_1,A\ve_2) = \alpha_{11}\alpha_{22} - \alpha_{12}\alpha_{21} = \det(A)\,,
$$
where $\det : {\rm GL}(V) \to \F_*$ is the determinant map.

In Proposition \ref{prop:azGL}, the two $V$-covariant quadrature systems $\Qo$ and $\Qo_A$ can thus be equivalent only if $\det A =1$. Introducing the {\em symplectic group} $\SL = \{A\in {\rm GL}(V)\mid \det(A) = 1\}$, the main consequence is that the set $\qq_{G_0\rtimes V}(\Omega,V)$ is empty whenever $G_0\nsubseteq \SL$. This important fact was already noticed in \cite[Section VI]{GiHoWo04}, where two $V$-covariant quadrature systems $\Po$ and $\Po'$ such that $\Po'=\Po_A$ for some element $A\in\SL$ are called {\em similar}. However, in general  similarity does not imply equivalence of quadratures, and it may happen that the set $\qq_{G_0\rtimes V}(\Omega,V)$ is empty also when $G_0\subseteq\SL$. Indeed, we have the following more precise statement.

\begin{proposition}\label{prop:Gquad-->Gmult}
Let $G_0\subseteq\SL$ be any subgroup. A quadrature system $\Po\in\qq_V(\Omega,V)$ is $(G_0\rtimes V)$-covariant if and only if its associated multiplier $m$ satisfies the equality
\begin{equation}\label{eq:def_inv_mult}
m_A = m \qquad \forall A\in G_0 \,.
\end{equation}
In this case, let $W$ be the Weyl system associated with $\Po$ and centered at the point $o\in\Omega$ such that ${\rm GL}(V)\cdot o = \{o\}$. Then, for any projective representation $U$ of $G_0$ associated with $\Po$, we have
\begin{equation}\label{eq:def_metap}
W(A\vv) = U(A)W(\vv)U(A)^* \qquad \forall \vv\in V,\ A\in G_0 \,.
\end{equation}
\end{proposition}
\begin{proof}
If $\Po\in\qq^m_V(\Omega,V,S)$ and $A\in G_0$, then $\Po_A\in\qq^{m_A}_V(\Omega,V,S_A) = \qq^{m_A}_V(\Omega,V,S)$ by Proposition \ref{prop:azGL}. Therefore, by Proposition \ref{prop:quad-->mult} the quadrature systems $\Po$ and $\Po_A$ are equivalent for all $A\in G_0$ if and only if \eqref{eq:def_inv_mult} holds. The second claim follows since the Weyl system associated with $\Po_A$ and centered at $o$ is the projective representation $W_A$ defined in \eqref{eq:defWA}. Hence, if $U(A)$ is a unitary operator intertwining $\Po$ with $\Po_A$ as in \eqref{eq:defU}, by Corollary \ref{cor:equiv} we have $W_A(\vv) = U(A)W(\vv)U(A)^*$ for all $\vv\in V$, which is \eqref{eq:def_metap}.
\end{proof}

Equation \eqref{eq:def_inv_mult} suggests to introduce and study the action of the group $\SL$ on the set of the Weyl multipliers, which transforms any multiplier $m$ into $m_A$ for all $A\in\SL$. In particular, it justifies the following definition.

\begin{definition}
If $G_0\subseteq\SL$ is a subgroup and $m$ is a Weyl multiplier, we say that $m$ is {\em $G_0$-invariant} if $m_A = m$ for all $A\in G_0$.
\end{definition}

Note that, for any symplectic form $S$, if $m\in\mm(V,S)$, then also $m_A\in\mm(V,S)$ for all $A\in\SL$. By Proposition \ref{prop:Gquad-->Gmult}, we are interested in the set of fixed points of $\mm(V,S)$ under the action of the group $\SL$ or some subgroup $G_0\subset \SL$. However, the next proposition shows that, when $G_0$ is too large, it may happpen that it actually has no fixed points in $\mm(V,S)$.

\begin{proposition}\label{prop:inv_mult}
Let $S$ be any symplectic form on $V$. There exists a $\SL$-invariant multiplier $m_{\rm inv}\in\mm(V,S)$ if and only if $p\neq 2$. In this case, $m_{\rm inv}$ is unique, and given by
$$
m_{\rm inv} (\vu,\vv) = \dual{2^{-1}\vv}{\vu} \qquad \forall \vu,\vv\in V \, .
$$
\end{proposition}
\begin{proof}
Suppose $m\in\mm(V,S)$ is $\SL$-invariant, and fix linearly independent vectors $\vu,\vv\in V$. If $\alpha,\beta,\gamma\in\F$, we then have
\begin{align*}
m(\vu,\gamma(\alpha\vu+\beta\vv)) & = \begin{cases}
m_A (\gamma\vu,\beta\vv) & \mbox{if }\gamma\beta\neq 0\\
1 & \mbox{if }\gamma\beta=0
\end{cases} \\
& = m (\gamma\vu,\beta\vv)\,,
\end{align*}
where $A\in\SL$ is given by
$$
A\vu = \gamma^{-1}\vu \qquad \qquad A\vv = \gamma\vv + \beta^{-1}\gamma\alpha \vu\,.
$$
For $\beta_1,\beta_2\in\F$, we also have
\begin{align*}
& m(\vu,(\beta_1+\beta_2)\vv) = m(\vu,(\beta_1+\beta_2)\vv)m(\beta_1\vv,\beta_2\vv) \\
& \qquad \qquad = m(\vu + \beta_1\vv,\beta_2\vv)m(\vu,\beta_1\vv) = m_B (\vu,\beta_2\vv)m(\vu,\beta_1\vv) \\
& \qquad \qquad = m(\vu,\beta_2\vv)m(\vu,\beta_1\vv) \,,
\end{align*}
with $B\in\SL$ as follows
$$
B\vu = \vu + \beta_1 \vv \qquad \qquad B\vv = \vv\,.
$$
Therefore, if $\vu_1,\vu_2\in\F\vu$, $\vv_1,\vv_2\in\F\vv$ and $\delta\in\F$,
\begin{align*}
& m(\vu,\delta(\vu_1+\vv_1 + \vu_2+\vv_2)) = m(\vu,\delta(\vv_1 + \vv_2)) \\
& \qquad\qquad = m(\vu,\delta\vv_1) m(\vu,\delta\vv_2) = m(\vu,\delta(\vu_1+\vv_1)) m(\vu,\delta(\vu_2+\vv_2))\\
& \qquad\qquad = m(\delta\vu,(\vu_1+\vv_1)) m(\delta\vu,(\vu_2+\vv_2)) \,.
\end{align*}
This relation with $\delta=1$ implies that $m(\vu,\cdot)\in\hat{V}$ for all $\vu$. Hence there is a unique vector $T(\vu)\in V$ such that $m(\vu,\vw) = \dual{\vw}{T(\vu)}$ for all $\vw\in V$. On the other hand, the same relation with $\vu_2+\vv_2=\vnull$ yields $m(\vu,\delta\vw) = m(\delta\vu,\vw)$, hence the map $T:V\to V$ satisfies $T\delta = \delta T$ for all $\delta\in\F$. Analogously, also $m(\cdot,\vv)\in\hat{V}$ for all $\vv$, which implies that $T(\vu_1+\vu_2) = T(\vu_1) + T(\vu_2)$ for all $\vu_1,\vu_2\in V$. Thus, $T$ is $\F$-linear. By $\SL$-invariance of $m$ it follows that $A^{-1}TA = T$ for all $A\in\SL$, implying that $T=\alpha I$ for some $\alpha\in\F$ \cite[Theorem 4.8]{Artin}. Finally,
$$
\dual{\vu}{2\alpha\vv} = \overline{\dual{\vv}{\alpha\vu}} \dual{\vu}{\alpha\vv} = \overline{m(\vu,\vv)} m(\vv,\vu) = \dual{\vu}{\vv} \,,
$$
so that $2\alpha = 1$. Therefore, $p$ must be odd and $\alpha = 2^{-1}$, so that $m(\vu,\vv) = \dual{\vv}{2^{-1}\vu} = \dual{2^{-1}\vv}{\vu}$.
\end{proof}

The next theorem is the main result of the section.

\begin{theorem}\label{teo:nonex_p2}
The set $\qq_{\SL\rtimes V}(\Omega,V)$ is nonempty if and only if $p\neq 2$. In this case, for any symplectic form $S\in\Sym$, the set of quadratures $\qq_V^{m_{\rm inv}} (\Omega,V,S) \equiv \qq_V (\Omega,V,S) \cap \qq_{\SL\rtimes V} (\Omega,V)$ is the unique $(\SL\rtimes V)$-invariant equivalence class in $\qq_V (\Omega,V,S)$
\end{theorem}
\begin{proof}
Immediate from Theorem \ref{teo:quad=mult} and Propositions \ref{prop:Gquad-->Gmult} and \ref{prop:inv_mult}.
\end{proof}

When $p\neq 2$, the distinguished role played by the $(\SL\rtimes V)$-invariant equivalence class $\qq_V^{m_{\rm inv}} (\Omega,V,S)$ inside $\qq_V (\Omega,V,S)$ was already observed in \cite[Section VI, after Equation (72)]{GiHoWo04} in the special case $\F=\Z_p$. Moreover, the nonexistence of $(\SL\rtimes V)$-covariant quadrature systems when $\F=\Z_2$ was also noticed in \cite[Section VIII]{GiHoWo04} (see also Section \ref{sec:spin1/2} below).

When $p=2$, even if there do not exist $(\SL\rtimes V)$-covariant quadrature systems, one can still find properly contained subgroups $G_0\subset\SL$ admitting $G_0$-invariant Weyl multipliers, so that the set $\qq_{G_0\rtimes V}(\Omega,V)$ is nonempty. A particularly important class of these subgroups is the subject of the next section.

\begin{remark}
Theorem \ref{teo:nonex_p2} has a counterpart for $(\SL\rtimes V)$-covariant Wigner functions (see \cite[Theorem 7]{Gr06} for the definition of Wigner functions that are covariant with respect to the group of the affine symplectic phase-space transformations). Indeed, the Wigner function analogue of the uniqueness statement in Theorem \ref{teo:nonex_p2} was established by Gross in the case $\F=\Z_p$ with $p$ odd \cite[Theorem 23]{Gr06}. Zhu recently extended Gross' uniqueness result to all finite fields with odd characteristic, and he also proved that there do not exist $(\SL\rtimes V)$-covariant Wigner functions in even characteristic \cite[Theorem 3]{Zh15}.
\end{remark}

\begin{remark}
When there exists a quadrature system $\Po\in\qq_{G_0\rtimes V}(\Omega,V)$ for some subgroup $G_0\subseteq\SL$, any Weyl system $W$ associated with $\Po$ can be enlarged to a projective representation of the whole semidirect product $G_0\rtimes V$ which is still associated with $\Po$. Indeed, this is done by defining the extension $\tilde{W}(A,\vv) = U(A)W(\vv)$ for all $(A,\vv)\in G_0\rtimes V$, where $U$ is any projective representation of $G_0$ associated with $\Po$. In particular, if $G_0 = \SL$ and $W$ is the Weyl system centered at the point $o\in\Omega$ such that ${\rm GL}(V) \cdot o = \{o\}$, then \eqref{eq:def_metap} implies that the representation $U$ is the {\em Weil} \cite{We64,Ho73,Ge77} or {\em metaplectic} \cite{Ne02} representation of the symplectic group $\SL$. In this case, the representation $\tilde{W}$ of the full semidirect product $\SL\rtimes V$ is known with the name of {\em Clifford group} in the physics literature \cite{DeDM03,Appleby05JMP} (see also \cite{Ap09} and the references therein; for an exhaustive mathematical description of the Clifford group, we refer to \cite{BoRoWa61I,BoRoWa61II}). Proposition \ref{prop:inv_mult} then reflects the well known difficulties which arise when one tries to define the Weil representation in characteristic $p=2$ \cite{BoRoWa61II,Bl93,GuHa12}.
\end{remark}

\begin{remark}\label{rem:phase}
In \cite[Appendix B]{GiHoWo04}, for every irreducible Weyl system $W$, in any characteristic $p$ and for all symplectic maps $A\in\SL$, the authors provide an explicit construction of a unitary operator $U_A$ satisfying the relation
\begin{equation}\label{eq:GiHoWo}
U_A W(\vv) U_A^* = a(A,\vv) W(A\vv)\qquad \forall \vv\in V \,,
\end{equation}
where $a:\SL\times V\to\T$ is a nontrivial phase function. In order to explain the origin of the operators $\{U_A\mid A\in\SL\}$ of \eqref{eq:GiHoWo}, observe that by Proposition \ref{prop:SvNmult} the multiplier $m_A$ of $W_A$ is equivalent to the multiplier $m$ of $W$, hence, if $a_A : V\to\T$ is any function intertwining $m_A$ with $m$, the irreducible Weyl systems $W$ and $W'_A = a_A W_A$ have the same multiplier. Setting $a(A,\vv) = a_A(\vv)$ for all $\vv\in V$, the existence of the unitary operator $U_A$ satisfiyng \eqref{eq:GiHoWo} then follows from Stone-von Neumann theorem (Proposition \ref{prop:SvNrep}).\\
However, when $W$ arises as a Weyl system associated with some quadrature system $\Po\in\qq^m_V(\Omega,V,S)$, we remark that in general such an operator $U_A$ does not intertwine the quadrature system $\Po$ with the transformed one $\Po_A$. Indeed, if $W$ is centered at some point of $\Omega$, the Weyl system $W'_A$, which is associated with $\Po_A$, needs not be centered at any point. Corollary \ref{cor:equiv} then does not apply to $\Po$ and $\Po_A$. Thus, the operator $U_A$ may not satisfy \eqref{eq:defU}, hence in general it is unrelated to the covariance properties of the quadrature system $\Po$. In particular, when $p=2$ the existence of the $U_A$'s constructed in \cite{GiHoWo04} is not in contradiction with Theorem \ref{teo:nonex_p2} above.\\
Nevertheless, the unitary operator $U_A$ still yields the range conjugacy relation $\ran\Po_A = U_A(\ran\Po)U_A^*$; that is, $U_A$ maps the maximal set of MUBs corresponding to the quadrature systems $\Po$ onto the one corresponding to $\Po_A$, if MUBs are regarded {\em as sets of unordered bases}. The reason of this fact is similar to the proof of Proposition \ref{prop:permutation}. Indeed, for all $D\in\ss$ the restriction $\left.a_A\right|_D$ is a character of $D$ by Proposition \ref{prop:SvNmult}, hence $\left.a_A\right|_D = \overline{\dual{\cdot}{\vv_{A,D}}}$ for some $\vv_{A,D}+D\in V/D$. Let $\Po'_A\in\qq_V(\Omega,V)$ be the quadrature system
$$
\Po'_A (x+D) = \Po_A(x+\vv_{A,D}+D) \qquad \forall x+D\in\Aff{\Omega} \,.
$$
If the Weyl system $W$ associated with $\Po$ is centered at the point $o\in\Omega$ such that ${\rm GL}(V)\cdot o = \{o\}$, then $W'_A$ is the Weyl system associated with $\Po'_A$ and still centered at $o$. Since $U_A$ intertwines $W$ with $W'_A$, it also intertwines $\Po$ with $\Po'_A$ by Corollary \ref{cor:equiv}. Therefore, $\ran\Po_A = \ran\Po'_A = U_A(\ran\Po)U_A^*$.\\
It is worth stressing that, unlike the Weil representation, the map $A\mapsto U_A$ is not guaranteed to be a projective representation of $\SL$. Actually, it is a projective representation if and only $U_{AB}^* U_A U_B$ is a complex scalar for all $A,B\in\SL$, which is equivalent to the condition
$$
U_A U_B W(\vv) U_B^* U_A^* = U_{AB} W(\vv) U_{AB}^* \qquad \forall \vv\in V,\, A,B\in\SL
$$
by irreducibility of $W$. Inserting \eqref{eq:GiHoWo} into this equation, we find that the function $a$ must satisfy the cocycle identity
$$
a(AB,\vv) = a(A,B\vv)a(B,\vv) \qquad \forall \vv\in V,\, A,B\in\SL\,.
$$
We have seen that this happens in characteristic $p\neq 2$ by choosing as $W$ any Weyl system whose multiplier $m$ is $\SL$-invariant and letting $U_A \equiv U(A)$ be the metaplectic representation. This is still true when $\F=\Z_2$ by making an appropriate choice of the opertators $\{U_A\mid A\in\SL\}$ (see Section \ref{sec:spin1/2} below). However, when $|\F|=2^r$ with $r\geq 2$, the existence of a cocycle $a:\SL\times V\to\T$ and a projective representation $A\mapsto U_A$ satisfying \eqref{eq:GiHoWo} is an open problem to our knowledge. In fact, \cite[Theorem 7]{BoRoWa61II} seems to give a strong indication against this possibility.
\end{remark}

We conclude this section with the following improvement of Proposition \ref{prop:permutation}.

\begin{theorem}\label{teo:permutation2}
Let $\Po_1$ and $\Po_2$ be any two $V$-covariant quadrature systems, with $\Po_i$ acting on the Hilbert space $\hh_i$. Then there exist a unitary operator $U:\hh_1\to\hh_2$, a map $A\in{\rm GL}(V)$ and, for all $D\in\ss$, a vector $\vu_D\in V$, such that
\begin{equation*}
\Po_2(\lf) = U\Po_1(A\lf+\vu_D)U^* \qquad \forall \lf\in\Af{D}{\Omega},\,D\in\ss\,.
\end{equation*}
\end{theorem}
\begin{proof}
Let $S_i$ be the symplectic form induced by $\Po_i$ on $(\Omega,V)$. Pick $A\in{\rm GL}(V)$ such that $S_2 = \det(A) S_1$, and let $\Po_2' = (\Po_2)_{A^{-1}}$. By Proposition \ref{prop:azGL}, we have $\Po_2'\in\qq_V(\Omega,V,S_1)$, hence there exist a unitary $U:\hh_1\to\hh_2$ and vectors $\vv_D\in V$ such that
\begin{equation*}
\Po_2'(\lf) = U\Po_1(\lf+\vv_D)U^* \qquad \forall \lf\in\Af{D}{\Omega},\,D\in\ss
\end{equation*}
by Proposition \ref{prop:permutation}. Going back to the quadrature system $\Po_2$ and substituting $\vu_D = \vv_{AD}$, we obtain the claim.
\end{proof}

The above result implies that the ranges of any two $V$-covariant quadrature systems are unitarily conjugated, regardless of the symplectic forms they induce on $(\Omega,V)$ (cf.~Remark \ref{rem:phase}, where the symplectic form $S$ was fixed). It should be stressed that this is a distinguished property of $V$-covariant quadratures, which does not extend to the noncovariant ones (see \cite{Kan12} for more details on quadrature systems whose ranges are not unitarily conjugated).

\section{Maximal nonsplit toruses and systems of rotated quadratures}\label{sec:nonsplit}

We now define the finite field analogues of the rotation group of the Euclidean plane $\R^2$, which have been first introduced in the context of MUBs by \cite{Chau05,SuWo07} and further studied in \cite{Su07,Ap09} (see also \cite{GuHaSo08,GuHaSo08bis} for applications in signal analysis). As it will be proved below, there exist quadrature systems in $\qq_V(\Omega,V)$ that are covariant with respect to such groups in all characteristics $p$ (even or odd).

\begin{definition}
An element $A\in\SL$ is {\em nonsplit} if $AD\neq D$ for all $D\in\ss$. A {\em nonsplit torus} is a cyclic subgroup of $\SL$ generated by a nonsplit element.
\end{definition}

An element $A\in\SL$ is nonsplit if and only if its characteristic polynomial
\begin{equation}\label{eq:charpol}
p_A(X) = \det(A-XI) = X^2 - {\rm tr}(A) X + 1
\end{equation}
has no solution in the field $\F$ (in the above expression, $I$ is the identity of $V$, and ${\rm tr}(A)$ is the trace of $A$).

In order to describe a nonsplit element $A\in\SL$ and the stucture of the torus it generates, it is useful to fix a symplectic basis of $V$ and represent $A$ as a unit determinant $2\times 2$ matrix with entries in $\F$. If $z$ and $\overline{z}$ are the two conjugate roots of $p_A$ in the quadratic extension $\tilde{\F}$ of $\F$, and $(\alpha_1,\alpha_2)^T$ and $(\overline{\alpha_1},\overline{\alpha_2})^T$ are the two eigenvectors of $A$ corresponding to the eigenvalues $z$ and $\overline{z}$, we have
$$
A = U \left(\begin{array}{cc} z & 0\\ 0 & \overline{z} \end{array}\right) U^{-1} \quad\text{with}\quad U = \left(\begin{array}{cc} \alpha_1 & \overline{\alpha_1} \\ \alpha_2 & \overline{\alpha_2} \end{array}\right) \,.
$$
In particular, the nonsplit torus generated by $A$ is the subgroup
$$
T_A = \left\{U \left(\begin{array}{cc} z^k & 0\\ 0 & \overline{z}^k \end{array}\right) U^{-1} \mid k\in\Z\right\}\,.
$$
Note that the commutant of $T_A$ in $\SL$ is the subgroup
$$
T_A' = \left\{U \left(\begin{array}{cc} z' & 0\\ 0 & \overline{z}' \end{array}\right) U^{-1} \mid z'\in\tilde{\F} \text{ and } z'\overline{z}' = 1\right\}\,.
$$
Since the set
\begin{equation}\label{eq:defM}
M = \{z'\in\tilde{\F} \mid z'\overline{z}' = 1\}
\end{equation}
is a cyclic subgroup of the multiplicative group $\tilde{F}_*=\tilde{F}\setminus\{0\}$ \cite[Theorem IV.1.9]{LanAlg}, the group $T_A'$ is cyclic: it is the {\em maximal} nonsplit torus containing $T_A$ (see \cite[Section 16.2]{Hump75} for the definition of toruses in general algebraic groups).

Maximal nonsplit toruses in $\SL$ are all the subgroups of the form
\begin{equation}\label{eq:maxnsgen}
T =
\left\{\frac{1}{\alpha_1\overline{\alpha_2} - \overline{\alpha_1}\alpha_2}\left(\begin{array}{ccc}
z_0^k \alpha_1\overline{\alpha_2} - \overline{z_0^k \alpha_1}\alpha_2 & \ & \overline{z_0^k \alpha_1}\alpha_1 - z_0^k \alpha_1\overline{\alpha_1} \\
z_0^k \alpha_2\overline{\alpha_2} - \overline{z_0^k \alpha_2}\alpha_2 & \ & \overline{z_0^k \alpha_2}\alpha_1 - z_0^k \alpha_2\overline{\alpha_1}
\end{array}\right) \mid k\in\Z\right\}
\end{equation}
where $\alpha_1,\alpha_2\in\tilde{\F}$ with $\alpha_1\overline{\alpha_2}\notin\F$ and $z_0$ is any generator of the cyclic group $M$ defined in \eqref{eq:defM}. A concrete example of a maximal nonsplit torus can be constructed in the following way: the symplectic matrix
\begin{equation}\label{eq:A_triang_nonsplit}
A = \left(\begin{array}{cc} z_0 + \overline{z_0} & 1\\ -1 & 0 \end{array}\right)
\end{equation}
has eigenvalues $z_0$ and $\overline{z_0}$, hence $T_A=T_A'$, that is, $T_A$ is a maximal nonsplit torus. This group $T_A$ corresponds to the choice $\alpha_1 = z_0$ and $\alpha_2 = -1$ in \eqref{eq:maxnsgen}, and it is the prototype of a maximal nonsplit torus, as every other maximal nonsplit torus is in the conjugacy class of $T_A$ in $\SL$ by \cite[Corollary A of Section 21.3]{Hump75}. In other words, by suitably choosing the symplecic basis of $V$, any maximal nonsplit torus can be put in the form $T=\{A^k\mid k\in\Z\}$ with $A$ given by \eqref{eq:A_triang_nonsplit}.

We now evaluate the order of a maximal nonsplit torus $T$. As $|T|=|M|$, this amounts to finding the order of $M$, that is, the kernel of the homomorphism $\phi:\tilde{\F}_* \to \F_*$ given by $\phi(z)=z\overline{z}$. In order to do it, observe first of all that $\phi(\F_*) = \F_*^2$, the group of the squares of $\F_*$. If $p=2$, then $\F_*^2 = \F_*$. Thus, $\phi$ is surjective, hence $|\tilde{\F}_*|/|M| = |\F_*|$, that is, $|M| = |\tilde{\F}_*|/|\F| = |\F|+1$. If $p\neq 2$, pick any element $\gamma\in\F_*\setminus\F_*^2$, and let $j,-j\in\tilde{\F}$ be its square roots. We have $\phi(\alpha+j) = \alpha^2-\gamma$ for all $\alpha\in\F$, hence $|\phi(\F+j)| = |\F^2| = |\F_*^2|+1$, which implies that $\F_*^2$ is a proper subgroup of $\phi(\tilde{\F}_*)$. Since $\F_*^2$ has index $2$ in $\F_*$, it follows that $\phi$ is surjective also in this case, hence $|M|=|\F|+1$ again.

Finally, we look at the action of a maximal nonsplit torus $T$ on the set of directions $\ss$. Since the intersection $M\cap\F_* = \{1,-1\}$, we see that $T$ contains exactly two split elements in characteristic $p\neq 2$, that is, $I$ and $-I$, while if $p=2$ it does not contain any nontrivial split element. Therefore, the stabilizer subgroup for the action of $T$ on $\ss$ is $\{I,-I\}$ if $p\neq 2$, and it is trivial if $p=2$. As $|\ss|=|T|$, this implies that when $p\neq 2$ the torus $T$ has two orbits in $\ss$ with $(|\F|+1)/2$ elements in each orbit, while it acts freely and transitively on $\ss$ when $p=2$.

We summarize the main points of the above discussion in the following proposition (cf.~\cite{SuWo07} in the case $p$ even, and \cite[Theorems 7 and 8]{Ap09} for $p$ odd).

\begin{proposition}\label{prop:maxnsplit}
There exist nonsplit toruses in $\SL$ for every characteristic $p$. Each nonsplit torus $T$ is contained in a uniquely determined maximal nonsplit torus, that is, its commutant $T'$ in $\SL$. A maximal nonsplit torus has order $|\F|+1$, and all maximal nonsplit toruses are conjugated in $\SL$. If $T$ is a maximal nonsplit torus, then its action on the set of directions $\ss$
\begin{itemize}
\item[-] is free and transitive if $p=2$;
\item[-] has two orbits with $(|\F|+1)/2$ elements in each orbit if $p\neq 2$.
\end{itemize}
\end{proposition}

If $T\subset\SL$ is a maximal nonsplit torus, the semidirect product $T\rtimes V$ is the finite analogue of the Euclidean group of the plane $\R^2$. According to this analogy, we say that any $\Po\in\qq_{T\rtimes V}(\Omega,V)$ is a {\em system of rotated quadratures}. To provide a further explanation of this terminology, let $A$ be a generator of $T$, fix a direction $D_{\oo}$ on each orbit $\oo$ of $T$ in $\ss$, and pick a vector $\vu\in V$ not belonging to any subspace $D_{\oo}$. Moreover, as usual denote by $o$ the point of $\Omega$ fixed by ${\rm GL}(V)$. Then, every affine line $\lf\in\Aff{\Omega}$ can be written as
$$
\lf = A^{k(\lf)}(o+\alpha(\lf)\vu+D_{\oo(\lf)}) \,,
$$
where the orbit $\oo(\lf)$ is uniquely determined by $\lf$, while the couple $(\alpha(\lf),k(\lf))\in\F\times\Z_{|\F|+1}$ is unambiguously defined if $p=2$, and it is unique up to the substitution $(\alpha(\lf),k(\lf)) \mapsto (-\alpha(\lf),k(\lf)+(|\F|+1)/2)$ if $p\neq 2$. Therefore, we have
\begin{equation}\label{eq:quad_rot}
\Po(\lf) = U(A)^{k(\lf)} W(\alpha(\lf)\vu) \Po(o+D_{\oo(\lf)}) W(\alpha(\lf)\vu)^* U(A)^{k(\lf)\,\ast} \,,
\end{equation}
where $U$ and $W$ are any projective representation of $T$ and any Weyl system associated with $\Po$, respectively (see Theorem \ref{teo:metap} below for the explicit form of $U$). The last formula shows that, when $p=2$ [respectively, when $p\neq 2$] every projection $\Po(\lf)$ can be obtained by unitary conjugation of one fixed projection $\Po(\lf_0)$ [resp., two fixed projections $\Po(\lf_1)$ and $\Po(\lf_2)$] by means of the representations $U$ and $W$, and it thus justifies the name of rotated quadratures for $\Po$.

The next easy result is the key fact for proving the existence of rotated quadratures in all field characteristics.

\begin{proposition}\label{prop:exnsplit}
If $T$ is a maximal nonsplit torus, then, for all $S\in\Sym$, there exists a $T$-invariant multiplier $m\in\mm(V,S)$.
\end{proposition}
\begin{proof}
If $p\neq 2$, it is enough to choose $m = m_{\rm inv}$. Otherwise, if $p=2$,
pick any $m_0\in\mm(V,S)$, and let $m=\prod_{A\in T} (m_0)_A$. Then $m$ is a multiplier of $V$, which clearly satisfies item (i) of Definition \ref{def:WHmul}. Since $\overline{(m_0)_A (\vu,\vv)} (m_0)_A (\vv,\vu) = \dual{\vu}{\vv} = (-1)^{\Tr{\sym{\vu}{\vv}}}$ for every $A\in T$, we have
$$
\overline{m(\vu,\vv)} m(\vv,\vu) = (-1)^{|T|\Tr{\sym{\vu}{\vv}}} = (-1)^{\Tr{\sym{\vu}{\vv}}} = \dual{\vu}{\vv}
$$
because $|T|=|\F|+1$ is odd. Therefore, also item (ii) of Definition \ref{def:WHmul} is satisfied by $m$, hence $m\in\mm(V,S)$. For all $B\in T$,
$$
m_B = \prod_{A\in T} (m_0)_{AB} = \prod_{A\in T} (m_0)_A = m \,,
$$
which shows that $m$ is $T$-invariant.
\end{proof}

We remark that in general, contrary to the case of invariant multipliers, $T$-invariant multipliers are not unique when $T$ is a maximal nonsplit torus (see Section \ref{sec:spin1/2} below for an example).

\begin{theorem}
For any characteristic $p$ and $S\in {\rm Sym}(V)$, if $T\subset\SL$ is a nonsplit torus, then the set $\qq_{T\rtimes V} (\Omega,V,S)$ is nonempty.
\end{theorem}
\begin{proof}
Since $\qq_{T_1\rtimes V} (\Omega,V,S)\subset\qq_{T_2\rtimes V} (\Omega,V,S)$ whenever $T_2\subset T_1$, it is not restrictive to assume that $T$ is maximal. In this case, the claim follows from Theorem \ref{teo:quad=mult} and Propositions \ref{prop:Gquad-->Gmult} and \ref{prop:exnsplit}.
\end{proof}

For any nonsplit torus $T\subset\SL$ and quadrature system $\Po\in\qq_{T\rtimes V} (\Omega,V)$, we now explicitely exhibit the projective representation $U$ of $T$ associated with $\Po$. Such a representation is the finite analogue of the oscillator representation of quantum homodyne tomography, and its effect is to rotate $\Po$ in different directions according to the action of $T$ on the set $\ss$, as described in formula \eqref{eq:quad_rot}. We stress again that no restriction is made on the characteristics $p$ of the field.

\begin{theorem}\label{teo:metap}
Let $T$ be a nonsplit torus, and suppose $\Po\in\qq_{T\rtimes V} (\Omega,V)$. Let $W$ be the Weyl system associated with $\Po$ and centered at the point $o\in\Omega$ such that ${\rm GL}(V)\cdot o = \{o\}$, and let $m$ be its Weyl multiplier. Then, the projective representation $U$ of $T$ associated with $\Po$ is given by
\begin{equation}\label{eq:formulaU}
U(A) = \frac{c(A)}{|\F|} \sum_{\vu\in V} m(\vu,(A-I)^{-1}\vu) W(\vu) \qquad \forall A\in T\setminus\{I\}\,,
\end{equation}
where $c(A)\in\T$ is an arbitrary phase factor.
\end{theorem}
\begin{proof}
By Proposition \ref{prop:SvNrep}, we can expand the operator $U(A)$ with respect to the basis $\{W(\vu)\mid\vu\in V\}$, that is,
$$
U(A) = \sum_{\vu\in V} \lam (\vu) W(\vu)
$$
for suitable coefficients $\lam (\vu)\in\C$. Equation \eqref{eq:def_metap} requires $U(A)W(\vv) = W(A\vv) U(A)$, which yields
$$
\sum_{\vu\in V} \lam (\vu) \overline{m(\vu,\vv)} W(\vu+\vv) = \sum_{\vu\in V} \lam (\vu) \overline{m(A\vv,\vu)} W(A\vv+\vu) \,.
$$
Comparing these two expansions, we have
$$
\lam (\vu-\vv) \overline{m(\vu-\vv,\vv)} = \lam (\vu-A\vv) \overline{m(A\vv,\vu-A\vv)} \,.
$$
Since $A-I$ is invertible, we can make the substitutions $\vx=\vu-\vv$ and $\vy=\vu-A\vv$. As at least one of the $\lam(\vy)$ is nonzero, in this way we obtain
\begin{align*}
\frac{\lam(\vx)}{\lam(\vy)} & = \frac{m(\vx,(A-I)^{-1}(\vx-\vy))}{m(A(A-I)^{-1}(\vx-\vy),\vy)} \,.
\end{align*}
By Proposition \ref{prop:Gquad-->Gmult}, the multiplier $m$ is $T$-invariant, hence 
\begin{align*}
\frac{\lam(\vx)}{\lam(\vy)} & = \frac{m(\vx,(A-I)^{-1}(\vx-\vy))}{m((A-I)^{-1}(\vx-\vy),A^{-1}\vy)} \\
& = \frac{m(\vx,(A-I)^{-1}(\vx-\vy)) m((A-I)^{-1}\vx,-(A-I)^{-1}\vy)}{m((A-I)^{-1}(\vx-\vy),A^{-1}\vy) m((A-I)^{-1}\vx,-(A-I)^{-1}\vy)} \\
& = \frac{m(A(A-I)^{-1}\vx,-(A-I)^{-1}\vy) m(\vx,(A-I)^{-1}\vx)}{m((A-I)^{-1}\vx,-A^{-1}(A-I)^{-1}\vy) m(-(A-I)^{-1}\vy,A^{-1}\vy)} \\
& = \frac{m(-A(A-I)^{-1}\vx,(A-I)^{-1}\vy) m(\vx,(A-I)^{-1}\vx)}{m(-A(A-I)^{-1}\vx,(A-I)^{-1}\vy) m(-(A-I)^{-1}\vy,A^{-1}\vy)} \\
& = \frac{m(\vx,(A-I)^{-1}\vx)}{m(-(A-I)^{-1}\vy,A^{-1}\vy)}
\end{align*}
where in the first and fourth equalities we used $T$-invariance of $m$ and the fact that $-I\in T$, and in the third one we employed the multiplier property of $m$. Then, being valid for all $\vx,\vy\in V$, this equation implies that, for all $\vu\in V$,
\begin{equation*}
m(\vu,(A-I)^{-1}\vu) = m(-(A-I)^{-1}\vu,A^{-1}\vu)
\end{equation*}
and
\begin{equation*}
\lam(\vu) = d(A) m(\vu,(A-I)^{-1}\vu) \,,
\end{equation*}
where $d(A)\in\C$ is a constant independent of $\vu$. From the unitarity condition $U(A)U(A)^* = \id$ it follows that
\begin{align*}
|\F| & = \tr{U(A)U(A)^*}\\
& = |d(A)|^2 \sum_{\vu,\vv\in V} m(\vu,(A-I)^{-1}\vu) \overline{m(\vv,(A-I)^{-1}\vv)} \tr{W(\vu)W(\vv)^*}\\
& = |d(A)|^2 |\F| \sum_{\vu\in V} |m(\vu,(A-I)^{-1}\vu)|^2 \\
& = |d(A)|^2|\F|^3
\end{align*}
hence $d(A) = c(A)/|\F|$, where $c(A)\in\T$ is a phase factor.
\end{proof}

Since a nonsplit torus is a cyclic group, the phase function $c:T\to\T$ appearing in \eqref{eq:formulaU} can always be chosen in such a way as to make $U$ an ordinary representation \cite[Proposition 2.1.1]{Karp}.

When $p\neq 2$ and $\Po\in\qq_V^{m_{\rm inv}} (\Omega,V,S)$, the previous theorem yields the expression
$$
U(A) = \frac{c(A)}{|\F|} \sum_{\vu\in V} \dual{2^{-1}\vu}{(A-I)^{-1}\vu} W(\vu)
$$
for the restriction of the metaplectic representation $U$ to the torus $T$. This formula should be compared with the analogous result first stated in \cite[Propositon 4]{BaIt86} for the particular case $\F=\Z_p$ with $p\in 4\Z-1$. See also \cite[Lemma 2]{Appleby05JMP} for the case $\F=\Z_p$, and \cite[Eqs.~(45), (47), (50), (52)]{Vourdas04}, \cite[item (1) in Proposition 4]{Vourdas08JFAA} for an arbitrary $\F$. In the latter case, an alternative construction is also provided in \cite{GuHa07}. (Note: references \cite{Vourdas04,Vourdas08JFAA,Appleby05JMP} allow to compute the expression of $U(A)$ when the generator $A$ has the form \eqref{eq:A_triang_nonsplit}).

If $p=2$, up to our knowledge the only analogues of \eqref{eq:formulaU} that can be found in the literature are \cite[Equation (13)]{Chau05} and the constructions described in \cite[Appendix B]{GiHoWo04} and \cite[Section 3.2.1]{Su07} (cf.~also the expression of a generic Clifford unitary given in \cite[Theorem 6]{DeDM03}). However, we remark that all these references provide an operator $U_A$ satisfying the weaker covariance condition \eqref{eq:GiHoWo} in place of \eqref{eq:def_metap}, hence not satisfying the covariance condition \eqref{eq:defU} in general (see the explanation in Remark \ref{rem:phase}).

\section{An example: the qubit case}\label{sec:spin1/2}

In this section, we apply the theory developed in the previous part to the simplest situation in which $\F=\Z_2$. Even in this elementary application, we will encounter all the special features of the case in characteristic $p=2$ that we described in the previous two sections. The next characterization of $\qq_V(\Omega,V)$ in the case $\F=\Z_2$ should be compared with the similar one obtained by different means in \cite[Section VI]{GiHoWo04}.

We use the explicit realization of the affine space $(\Omega,V)$ described in Remark \ref{rem:explicit1}, that is, $\Omega = V = \Z_2^2$. There exists a unique symplectic form $S$ on $V$, which is given by $\sym{\ve_1}{\ve_2}=\sym{\ve_2}{\ve_1}=1$, where $\{\ve_1=(1,0)^T,\,\ve_2=(0,1)^T\}$ is the standard basis of $\Z_2^2$. The $3$ directions of $\Omega$ are the subspaces
$$
\ss = \{\F\ve_1,\, \F\ve_2,\, \F(\ve_1+\ve_2)\} \,,
$$
and the corresponding sets of parallel lines in $\Omega$ are
\begin{align*}
\Af{\F\ve_1}{\Omega} & = \left\{\left\{(0,0)^T,(1,0)^T\right\} , \, \left\{(0,1)^T,(1,1)^T\right\}\right\}\\
\Af{\F\ve_2}{\Omega} & = \left\{\left\{(0,0)^T,(0,1)^T\right\} , \, \left\{(1,0)^T,(1,1)^T\right\}\right\}\\
\Af{\F(\ve_1+\ve_2)}{\Omega} & = \left\{\left\{(0,0)^T,(1,1)^T\right\} , \, \left\{(0,1)^T,(1,0)^T\right\}\right\} \,.
\end{align*}

The following projective representation $W$ of $V$ in the Hilbert space $\hh = \C^2$ is a Weyl system for the symplectic space $(V,S)$
$$
W(\ve_1) = \sigma_1\,,\qquad W(\ve_2) = \sigma_2\,,\qquad W(\ve_1+\ve_2) = \sigma_3 \,,
$$
where $\sigma_1$, $\sigma_2$, $\sigma_3$ are the three Pauli matrices, with
\begin{gather*}
\sigma_i^2 = \id\,, \qquad \sigma_i \sigma_j = - \sigma_j \sigma_i \mbox{ if $i\neq j$} \\
\sigma_1 \sigma_2 = i\sigma_3\,,\qquad \sigma_2 \sigma_3 = i\sigma_1\,,\qquad \sigma_3 \sigma_1 = i\sigma_2\,.
\end{gather*}
The multiplier of $W$ is
\begin{align*}
m(\ve_1,\ve_2) & = m(\ve_1+\ve_2,\ve_1) = m(\ve_2,\ve_1+\ve_2) = -i\\
m(\ve_2,\ve_1) & = m(\ve_1,\ve_1+\ve_2) = m(\ve_1+\ve_2,\ve_2) = i\\
m(\ve_1,\ve_1) & = m(\ve_2,\ve_2) = m(\ve_1+\ve_2,\ve_1+\ve_2) = 1 \,.
\end{align*}

By Theorem \ref{teo:quad=mult}, there exists an equivalence class $\qq^m_V(\Omega,V,S)$ of $V$-covariant quadrature systems having associated multiplier $m$, and such a class is unique. Choosing the origin $o=(0,0)^T\in\Omega$, formula \eqref{eq:defPbis} yields the following explicit expression of an element $\Po\in\qq^m_V(\Omega,V,S)$
\begin{align}\label{eq:quadZ2}
\begin{aligned}
\Po(\{(0,0)^T,(1,0)^T\}) & = \frac{1}{2}(\id+\sigma_1) \qquad \Po(\{(0,1)^T,(1,1)^T\}) = \frac{1}{2}(\id-\sigma_1)\\
\Po(\{(0,0)^T,(0,1)^T\}) & = \frac{1}{2}(\id+\sigma_2) \qquad \Po(\{(1,0)^T,(1,1)^T\}) = \frac{1}{2}(\id-\sigma_2)\\
\Po(\{(0,0)^T,(1,1)^T\}) & = \frac{1}{2}(\id+\sigma_3) \qquad \Po(\{(1,0)^T,(0,1)^T\}) = \frac{1}{2}(\id-\sigma_3)
\end{aligned}
\end{align}

The multiplier $m$ and its complex conjugate $\overline{m}$ are the only two elements in the set $\mm(V,S)$. Therefore, the two equivalence classes $\qq_V^m(\Omega,V,S)$ and $\qq_V^{\overline{m}}(\Omega,V,S)$ are the only two classes in $\qq_V(\Omega,V,S)\equiv\qq_V(\Omega,V)$. A quadrature system $\Po'\in\qq_V^{\overline{m}}(\Omega,V,S)$ can be obtained by interchanging the Pauli matrices $\sigma_1$ and $\sigma_2$ in the definition \eqref{eq:quadZ2} of $\Po$.

The symplectic group $\SL$ is the semidirect product of an order $3$ normal cyclic subgroup and an order $2$ group. More precisely, let $R,F\in\SL$ be defined by
\begin{align*}
R\ve_1 & = \ve_1+\ve_2\,,\quad R\ve_2 = \ve_1\,,\qquad F\ve_1 = \ve_2\,,\quad F\ve_2 = \ve_1\,,
\end{align*}
and let $T$ and $H$ be the cyclic subgroups generated by $R$ and $F$, respectively. Then, $|T|=3$, $|H|=2$, $T$ is normal in $\SL$ and $\SL$ is the semidirect product $H\rtimes T$, where the action of $H$ on $T$ is given by $FRF^{-1} = R^{-1}$. Moreover, $T$ is the unique maximal nonsplit torus in $\SL$.

It is easy to see that both multipliers $m$ and $\overline{m}$ are $T$-invariant. Therefore, $\Po,\Po'\in\qq_{T\rtimes V}(\Omega,V)$. However, according to Theorem \ref{teo:nonex_p2}, neither $\Po$ nor $\Po'$ is $(\SL\rtimes V)$-covariant. Indeed, one immediately checks that actually $\Po_F = \Po'$, that is, the quadrature systems $\Po$ and $\Po'$ are {\em similar} in the terminology of \cite{GiHoWo04}.

By Theorem \ref{teo:metap}, a projective representation $U$ of $T$ satisfying \eqref{eq:defU} is given by
\begin{align*}
U(R) & = \frac{c(R)}{2} \sum_{\vu\in V} m(\vu,R\vu) W(\vu) = \frac{c(R)}{2} [\id + i (\sigma_1 + \sigma_2 + \sigma_3) ]\\
& = c(R) \, \e^{i\frac{\pi}{3} \vec{n}\cdot\vec{\sigma}}
\end{align*}
where we used the fact that $(R-I)^{-1} = R$ and denoted
$$
\vec{n} = \frac{1}{\sqrt{3}} (1,1,1) \qquad \vec{\sigma} = (\sigma_1,\sigma_2,\sigma_3)\,.
$$
In order to determine the phase factor $c(R)$ which turns $U$ into an ordinary representation, we impose the condition $\id = U(R^3) = U(R)^3 = -c(R)^3\id$, which implies that $c(R)$ must be any cubic root of $-1$.

We conclude this section observing that the representation $U$ can be extended to a projective representation of the whole group $\SL$ in such a way that the covariance relation
$$
U(A) W(\vv) U(A)^* = a(A,\vv) W(A\vv)\qquad \forall \vv\in V,\, A\in\SL
$$
is satisfied for some choice of the cocycle $a : \SL\times V \to \T$ (see Remark \ref{rem:phase}). Indeed, this is done by defining the unitary operator
$$
U(F) = \e^{i\frac{\pi}{2} \vec{m}\cdot\vec{\sigma}} \quad \text{with} \quad \vec{m} = \frac{1}{\sqrt{2}} (1,-1,0) \,,
$$
which is such that
$$
U(F)^2 = -\id\,, \qquad U(F)U(R)U(F^{-1}) = -U(FRF^{-1})
$$
and
$$
U(F)W(\vv)U(F)^* = -W(F\vv) \qquad \forall \vv\in V\,.
$$
Note that, as expected, the operator $U(F)$ does not intertwine $\Po$ with $\Po_F$, since
\begin{align*}
U(F)\Po(\{(0,0)^T,(1,0)^T\})U(F)^* & = \Po(\{(1,0)^T,(1,1)^T\}) \\
& \neq \Po_F(\{(0,0)^T,(1,0)^T\}) \\
U(F)\Po(\{(0,1)^T,(1,1)^T\})U(F)^* & = \Po(\{(0,0)^T,(0,1)^T\}) \\
& \neq \Po_F(\{(0,1)^T,(1,1)^T\}) \\
U(F)\Po(\{(0,0)^T,(1,1)^T\})U(F)^* & = \Po(\{(0,1)^T,(1,0)^T\}) \\
& \neq \Po_F(\{(0,0)^T,(1,1)^T\}) \,.
\end{align*}

\section{Conclusions}

We have classified all the equivalence classes of unitarily conjugated $V$-covariant MUBs for an affine space $(\Omega,V)$ over a finite field $\F$. We have shown that such classes are in one-to-one correspondence with a special family of multipliers of $V$, which we called Weyl multipliers. By studying the invariance properties of Weyl multipliers with respect to different subgroups $G_0\subseteq {\rm GL}(V)$, we have been able to characterize $(G_0\rtimes V)$-covariant MUBs for all possible choices of $G_0$. In particular, we have found that $(\SL\rtimes V)$-covariant MUBs exist if and only if the field $\F$ has characteristic $p\neq 2$, and in this case their equivalence class is unique. In characteristic $p=2$, however, $(G_0\rtimes V)$-covariance can be still achieved if $G_0$ is a maximal nonsplit torus in $\SL$, and we used this fact to construct covariant MUBs that are the finite analogues of the rotated quadrature observables in quantum homodyne tomography.

Our classification employed the alternative description of MUBs by means of their associated families of spectral resolutions, which we called quadrature systems in the paper. As a remarkable fact, it turned out that the ranges of all $V$-covariant quadrature systems are unitarily conjugated. This peculiarity singles out $V$-covariant MUBs as very special objects in the whole set of maximal MUBs. Moreover, it also shows that their different symmetry properties are a mere effect of the choice of inequivalent labelings with phase-space lines. In other words, they are exclusively the result of different orderings of the same sets of bases.

\section*{Acknowledgements.}

The authors thank Paul Busch and Markus Grassl for bringing Zhu's recent paper \cite{Zh15} to their attention. JS and AT acknowledge financial support from the Italian Ministry of Education, University and Research (FIRB project RBFR10COAQ).

\appendix

\section{Projective representations}\label{app:projrep}

In this appendix, $G$ is a finite group with additive composition law. We recall that a (unitary) {\em multiplier} of $G$ is a map $m:G\times G\to \T$ such that
\begin{equation*}
m(g_1 + g_2 , g_3) m(g_1 , g_2) = m(g_1 , g_2 + g_3) m(g_2 , g_3) \qquad \forall g_1,g_2,g_3 \in G \, .
\end{equation*}
The set of multipliers of $G$ forms a group under pointwise multiplication and inverse. The multiplier $m$ is {\em exact} if there exists a function $a:G\to \T$ such that $m(g_1,g_2) = \overline{a(g_1) a(g_2)} a(g_1 + g_2)$ for all $g_1,g_2\in G$. Two multipliers $m_1$, $m_2$ of $G$ are {\em equivalent} if $\overline{m_1} m_2$ is exact. When $m_1$ and $m_2$ are equivalent, any function $a:G\to \T$ such that $m_2(g_1,g_2) = \overline{a(g_1) a(g_2)} a(g_1 + g_2) m_1(g_1,g_2)$ for all $g_1,g_2\in G$ is said to {\em intertwine} the multiplier $m_1$ with $m_2$. In this case, the function $a$ is uniquely determined up to multiplication by a homomorphism $\chi:G\to\T$; that is, if also the function $a':G\to\T$ intertwines $m_1$ with $m_2$, then $a'(g) = \chi(g)a(g)$ for all $g$, where $\chi(g_1+g_2) = \chi(g_1)\chi(g_2)$ for all $g_1,g_2\in G$.

For a fixed a multiplier $m$ of $G$, a (unitary) {\em projective representation} of $G$ in the Hilbert space $\hh$ with multiplier $m$ is a map $R:G\to \uelle{\hh}$ such that
$$
R(g_1 + g_2) = m(g_1 , g_2) R(g_1) R(g_2) \qquad \forall g_1,g_2\in G \, .
$$
Note that $R(0) = \overline{m(0,0)}\id$, since $R(0)^2 = \overline{m(0,0)} R(0)$ and hence $R(0)(R(0)-\overline{m(0,0)}\id) = 0$. This easily implies the relation $R(g)^* = m(g,-g)m(0,0)R(-g)$ for all $g\in G$.

Actually, the projective representation $R$ is an {\em ordinary representation} if $m=1$. As a consequence, there exists a function $a:G\to \T$ such that the projective representation $R_a = aR$ is an ordinary representation if and only if $m$ is exact.

Two projective representations $R_1$ and $R_2$ acting on the Hilbert spaces $\hh_1$ and $\hh_2$, respectively, are {\em equivalent} if there exists a unitary map $U:\hh_1\to\hh_2$ such that $R_2(g) = UR_1(g)U^*$ for all $g\in G$. We remark that equivalent projective representations must have the same multiplier.

We say that the projective representation $R$ is {\em irreducible} if it leaves invariant no nontrivial subspace of $\hh$. Clearly, irreducibility is a property of the whole equivalence class of $R$.
The next result follows from \cite[Theorem 7.5]{GQT85}. For completeness, we report a shorter proof adapted to the present simplified setting.

\begin{proposition}\label{prop:exrep}
Suppose $m$ is a multiplier of $G$. Then there exists an irreducible projective representation of $G$ with multiplier $m$.
\end{proposition}
\begin{proof}
As in the proof of \cite[Theorem 7.5]{GQT85}, for $g\in G$ we define the following linear map $R(g) : \ell^2 (G) \to \ell^2 (G)$, where $\ell^2 (G)$ is the usual Hilbert space of complex functions on $G$
$$
[R(g)\phi] (x) = \overline{m(x,g)} \phi(x+g) \qquad \forall \phi\in\ell^2 (G) \, .
$$
It is immediately checked that $R$ is a projective representation of $G$ in $\ell^2 (G)$ with multiplier $m$. Restricting it to some irreducible subspace of $\ell^2 (G)$ we get the claim.
\end{proof}

The next easy sufficient condition for a projective representation to have an exact multiplier turns out to be quite useful.

\begin{proposition}\label{prop:trivmult}
Suppose $R$ is a projective representation of $G$ in the Hilbert space $\hh$. If for some $1$-dimensional subspace $\hh_0\subseteq\hh$ one has $R(g)\hh_0 = \hh_0$ for all $g\in G$, then the multiplier of $R$ is exact.
\end{proposition}
\begin{proof}
If $\phi$ is a nonzero vector in $\hh_0$, then for all $g\in G$ there exists a scalar $a(g)\in\T$ such that $R(g)\phi = a(g)\phi$. Therefore,
\begin{align*}
a(g_1+g_2)\phi & = R(g_1+g_2)\phi = m(g_1,g_2)R(g_1)R(g_2)\phi \\
& = m(g_1,g_2)a(g_1)a(g_2)\phi \,,
\end{align*}
that is, $m(g_1,g_2)=\overline{a(g_1)a(g_2)}a(g_1+g_2)$.
\end{proof}

The following is \cite[Lemma 7.2]{Klep65}. Again, we add a simpler proof for the reader's convenience.

\begin{proposition}\label{prop:abelrep}
Suppose $G$ is an abelian group, and let $m$ be a multiplier of $G$. If $m(g_1,g_2) = m(g_2,g_1)$ for all $g_1,g_2\in G$, then $m$ is exact.
\end{proposition}
\begin{proof}
Let $R$ be an irreducible projective representation of $G$ with multiplier $m$. It exists by Proposition \ref{prop:exrep}. Define a group law on the set $G_m:=G\times\T$ by
\[
(g_1,z_1)(g_2,z_2)=(g_1+g_2,z_1z_2\overline{m(g_1,g_2)})
\]
Since $m$ is symmetric, $G_m$ is abelian.
It is well known that $R$ lifts to an ordinary representation $R_m$ of $G_m$ as follows:
\[
R_m(g,z)=z\,R(g)\,.
\]
Clearly $R_m$ is irreducible (because it has the same commutant as $R$) hence it is $1$-dimensional. So also $R$ is $1$-dimensional, and $m$ is exact by Proposition \ref{prop:trivmult}.
\end{proof}

We conclude this appendix with the following alternative version of \cite[Lemma 7.1]{Klep65}. Before its statement, we recall that a {\em bicharacter} of an abelian group $G$ is a map $b:G\times G\to\T$ such that $b(g,\cdot)$ and $b(\cdot,g)$ are characters (i.e., $1$-dimensional homomorphisms) of $G$ for all $g\in G$. The bicharacter $b$ is {\em antisymmetric} if $b(g_1,g_2) = \overline{b(g_2,g_1)}$ for all $g_1,g_2\in G$.

\begin{proposition}\label{prop:bichar}
Suppose $G$ is abelian, and let $R$ be a projective representation of $G$. Then there exists a unique antisymmetric bicharacter $b$ of $G$ such that
\begin{equation}\label{eq:bichar}
R(g_1)R(g_2)R(g_1)^* = b(g_1,g_2) R(g_2) \qquad \forall g_1,g_2\in G \,.
\end{equation}
\end{proposition}
\begin{proof}
Since $R$ is a projective representation, \eqref{eq:bichar} holds for $b(g_1,g_2) = m(g_1,-g_1)m(0,0)\overline{m(g_1,g_2) m(g_1+g_2,-g_1)}$, where $m$ is the multiplier of $R$. Clearly, $b(g_1,g_2)\in\T$. Moreover, for fixed $g\in G$, the map $b(\cdot,g): G\to \T$ is a character, since
\begin{align*}
& b(g_1+g_2,g) R(g) = R(g_1+g_2)R(g)R(g_1+g_2)^* \\
&\qquad\quad = R(g_1)R(g_2)R(g)R(g_1)^*R(g_2)^* = b(g_1,g)b(g_2,g) R(g) \,.
\end{align*}
Equation \eqref{eq:bichar} also reads
$$
R(g_2)R(g_1)R(g_2)^* = \overline{b(g_1,g_2)} R(g_1) \qquad \forall g_1,g_2\in G \,,
$$
hence by comparison $b(g_2,g_1) = \overline{b(g_1,g_2)}$. As a consequence, for all $g\in G$ also the map $b(g,\cdot): G\to \T$ is a character, and the bicharacter $b$ is antisymmetric.
\end{proof}

\section{A Weyl multiplier in characteristic $p=2$}\label{app:mult}

The demonstration of the existence of Weyl multipliers provided in the proof of Proposition \ref{prop:SvNmult} is nonconstructive. Indeed, although it is easy to find an example of a Weyl multiplier in characteristic $p\neq 2$ (see Example \ref{ex:multodd}), the same task is more involved when $p=2$. In this appendix, we are going to fill this gap and explicitely exhibit a Weyl multiplier $m$ for the symplectic space $(V,S)$ when the characteristic of the scalar field $\F$ is even. Moreover, given a maximal nonsplit torus $T\subset\SL$, we will also make use of $m$ to construct a $T$-invariant multiplier in $\mm(V,S)$.

The construction is based on the fact that, in characteristic $p=2$, for all $\alpha\in\F_*$ there exists a linear basis $\{\eps^\alpha_1 , \eps^\alpha_2 ,\ldots , \eps^\alpha_n\}$ of $\F$ over $\Z_2$ such that $\Tr{(\alpha \eps^\alpha_i \eps^\alpha_j)} = \delta_{i,j}$ for all $i,j\in\{1,2,\ldots,n\}$. Indeed, by \cite[Theorem 4]{SeLe80} (see also \cite{Ima83,JuMeVa90}), there exists a linear basis $\{\omega_1 , \omega_2 , \ldots , \omega_n\}$ of $\F$ over $\Z_2$ such that $\Tr{(\omega_i \omega_j)} = \delta_{i,j}$ for all $i,j\in\{1,2,\ldots,n\}$. Since $p=2$, we have $\alpha = \gamma^2$ for some $\gamma\in\F_*$. Defining $\eps^\alpha_i = \gamma^{-1}\omega_i$, we then get a basis with the claimed property.

The square map $z\mapsto z^2$ is well defined from the field $\Z_2$ to the ring $\Z_4$. It follows that also the map $z\mapsto i^{z^2}$ is well defined from $\Z_2$ to $\T$. As $(z + t)^2 = z^2 + 2 zt + t^2$, we have $i^{(z + t)^2} = i^{z^2} i^{t^2} (-1)^{zt}$ for all $z,t\in\Z_2$.

For all $\alpha\in\F_*$, we use this fact to define the function $c_\alpha:\F\to\T$ with
$$
c_\alpha\left(\sum_{i=1}^n z_i \eps^\alpha_i\right) = \prod_{i=1}^n i^{z_i^2} \qquad \forall z_1,\ldots,z_n\in\Z_2 \,.
$$
Clearly, $c_\alpha(0) = 1$, and, by the previous paragraph,
$$
c_\alpha(\gamma+\delta) = c_\alpha(\gamma)c_\alpha(\delta)(-1)^{\Tr{(\alpha\gamma\delta)}} \qquad \forall\gamma,\delta\in\F \,.
$$

Now, choose a symplectic basis $\{\ve_1,\ve_2\}$ of $V$, and by means of it define the following multiplier $m_0$ of $V$
$$
m_0 (\alpha_1\ve_1 + \alpha_2\ve_2 \, , \, \beta_1\ve_1 + \beta_2\ve_2) = (-1)^{\Tr{(\beta_1 \alpha_2)}} \, .
$$
It is clear that $m_0$ satisfies item (i) of Definition \ref{def:WHmul}. We are going to find a multiplier $m$ equivalent to $m_0$ and fulfilling also condition (ii)  of Definition \ref{def:WHmul}. To this aim, for all $\alpha\in\F_*$ we fix a vector $\vv_\alpha = \ve_1 + \alpha\ve_2\in V$, and observe that, since $\ss = \{\F\ve_1,\F\ve_2\}\cup\{\F\vv_\alpha\mid\alpha\in\F_*\}$, a function $a:V\to\T$ can be defined as follows
$$
a(\vu) = \begin{cases}
c_\alpha(\gamma) & \mbox{ if } \vu = \gamma\vv_\alpha \mbox{ for some } \alpha \in\F_* \\
1 & \mbox{ if } \vu \in \F\ve_1 \cup \F\ve_2 
\end{cases} \, .
$$
We then claim that the multiplier $m$ of $V$ given by
$$
m(\vu,\vv) = \overline{a(\vu)a(\vv)} a(\vu+\vv) m_0 (\vu,\vv) \qquad \forall\vu,\vv\in V
$$
satisfies item (ii) of Definition \ref{def:WHmul}. Indeed, if $\vd_1,\vd_2\in\F\ve_1$ or $\vd_1,\vd_2\in\F\ve_2$, then $m(\vd_1,\vd_2) = 1$ by definitions. If instead $\vd_1,\vd_2\in\F\vv_\alpha$, with $\vd_i = \gamma_i\vv_\alpha$, then
\begin{align*}
m(\vd_1,\vd_2) & = \overline{c_\alpha(\gamma_1)c_\alpha(\gamma_2)}c_\alpha(\gamma_1 + \gamma_2)m_0(\gamma_1\ve_1 + \gamma_1\alpha\ve_2 \,,\, \gamma_2\ve_1 + \gamma_2\alpha\ve_2) \\
& = \overline{c_\alpha(\gamma_1)c_\alpha(\gamma_2)}c_\alpha(\gamma_1)c_\alpha(\gamma_2)(-1)^{\Tr{(\alpha\gamma_1\gamma_2)}} (-1)^{\Tr{(\alpha\gamma_1\gamma_2)}} = 1 \,.
\end{align*}
Therefore, $m\in\mm(V,S)$.

Finally, if $T\subset\SL$ is a maximal nonsplit torus, by the proof of Proposition \ref{prop:exnsplit} the multiplier $m'$ given by
$$
m'(\vu,\vv) = \prod_{A\in T} m(A\vu,A\vv) \qquad \forall\vu,\vv\in V
$$
is $T$-invariant element in $\mm(V,S)$.


\begin{thebibliography}{99}

\bibitem{Schwinger60}
J.~Schwinger,
\newblock Unitary operator bases,
\newblock {\em Proc.~Nat.~Acad.~Sci.~U.S.A.} {\bf 46} (1960) 570--579.

\bibitem{BaBoRoVa02}
S.~Bandyopadhyay, P.O.~Boykin, V.~Roychowdhury and F.~Vatan,
\newblock A new proof for the existence of mutually unbiased bases,
\newblock {\em Algorithmica} {\bf 34}(4) (2002) 512--528.

\bibitem{Ho05}
R.~Howe,
\newblock Nice error bases, mutually unbiased bases, induced representations, the Heisenberg group and finite geometries,
\newblock {\em Indag.~Math.~(N.S.)} {\bf 16}(3-4) (2005) 553--583.

\bibitem{Wootters87}
W.K.~Wootters
\newblock A Wigner-function formulation of finite-state quantum mechanics,
\newblock {\em Ann. Physics} {\bf 176} (1987) 1--21.

\bibitem{GiHoWo04}
K.S.~Gibbons, M.J.~Hoffman and W.K. Wootters,
\newblock Discrete phase space based on finite fields,
\newblock {\em Phys. Rev. A} {\bf 70} (2004) 062101.

\bibitem{Vourdas04}
A.~Vourdas,
\newblock Quantum systems with finite Hilbert space,
\newblock {\em Rep.~Progr.~Phys.} {\bf 67} (2004) 267--320.

\bibitem{ApBeCh08}
D.M.~Appleby, I.~Bengtsson and S.~Chaturvedi,
\newblock Spectra of phase point operators in odd prime dimensions and the extended Clifford group,
\newblock {\em J.~Math.~Phys.} {\bf 49}(1) (2008) 012102.

\bibitem{AusTol79}
L.~Auslander and R.~Tolimieri,
\newblock Is computing with the finite Fourier transform pure or applied mathematics?,
\newblock {\em Bull.~Amer.~Math.~Soc.~(N.S.)} {\bf 1}(6) (1979) 847--897.

\bibitem{BaIt86}
R.~Balian and C.~Itzykson,
\newblock Observations sur la m\'ecanique quantique finie,
\newblock {\em C.~R.~Acad.~Sci.~Paris S\'er.~I Math.} {\bf 303}(16) (1986) 773--778.

\bibitem{Var95}
V.S.~Varadarajan,
\newblock Variations on a theme of Schwinger and Weyl,
\newblock {\em Lett. Math.~Phys.} {\bf 34}(3) (1995) 319--326.

\bibitem{Vourdas97}
A.~Vourdas,
\newblock Phase space methods for finite quantum systems,
\newblock {\em Rep.~Math.~Phys.} {\bf 40} (1997) 367--371.

\bibitem{We64}
A.~Weil,
\newblock Sur certains groupes d'op\'erateurs unitaires,
\newblock {\em Acta Math.} {\bf 111} (1964) 143--211.

\bibitem{Ho73}
R.E.~Howe,
\newblock On the character of Weil's representation,
\newblock {\em Trans.~Amer.~Math.~Soc.} {\bf 177} (1973) 287--298.

\bibitem{Ge77}
P.~G{\'e}rardin,
\newblock Weil representations associated to finite fields,
\newblock {\em J.~Algebra} {\bf 46}(1) (1977) 54--101.

\bibitem{Ne02}
M.~Neuhauser,
\newblock An explicit construction of the metaplectic representation over a finite field,
\newblock {\em J.~Lie Theory} {\bf 12}(1) (2002) 15--30.

\bibitem{Kan12}
W.M.~Kantor,
\newblock MUBs inequivalence and affine planes,
\newblock {\em J.~Math.~Phys.} {\bf 53}(3) (2012) 032204.

\bibitem{Gr06}
D.~Gross,
\newblock Hudson's theorem for finite-dimensional quantum systems,
\newblock {\em J.~Math. Phys.} {\bf 47}(12) (2006) 122107.

\bibitem{Zh15}
H.~Zhu,
\newblock Permutation symmetry determines the discrete Wigner function,
\newblock arXiv:1504.03773 (2015).

\bibitem{Su07}
D.M.~Sussman,
\newblock Minimum-uncertainty states and rotational invariance in discrete phase space,
\newblock Thesis, William College (2007).

\bibitem{Ivanovic81}
I.D.~Ivanovi\'c,
\newblock Geometrical description of quantal state determination,
\newblock {\em J.~Phys. A: Math.~Gen.} {\bf 14}(12) (1981) 3241--3245.

\bibitem{LanAlg}
S.~Lang,
\newblock {\em Algebra}, 3rd edition,
\newblock Graduate Texts in Mathematics, No.~211 (Springer-Verlag, New York, 2002).

\bibitem{WoFi89}
W.K.~Wootters and B.D.~Fields,
\newblock Optimal state-determination by mutually unbiased measurements,
\newblock {\em Ann.~Physics} {\bf 191} (1989) 363--381.

\bibitem{Zel91}
E.I.~Zelenov,
\newblock $p$-adic quantum mechanics and coherent states,
\newblock {\em Teoret.~Mat.~Fiz.} {\bf 86}(2) (1991) 210--220, English translation in {\em Theoret.~and Math.~Phys.} {\bf 86}(2)  (1991) 143–151.

\bibitem{BSZ92}
J.C.~Baez, I.E.~Segal and Z.-F.~Zhou,
\newblock {\em Introduction to algebraic and constructive quantum field theory},
\newblock Princeton Series in Physics (Princeton University Press, Princeton, NJ, 1992).

\bibitem{DiHuVa99}
T.~Digernes, E.~Husstad and V.S.~Varadarajan,
\newblock Finite approximation of Weyl systems,
\newblock {\em Math.~Scand.} {\bf 84}(2) (1999) 261--283.

\bibitem{PaZa88}
J.~Patera and H.~Zassenhaus,
\newblock The Pauli matrices in $n$ dimensions and finest gradings of simple Lie algebras of type $A_{n-1}$,
\newblock {\em J.~Math.~Phys.} {\bf 29}(3) (1988) 665--673.

\bibitem{DuEnBe10}
T.~Durt, B.-G.~Englert, I.~Bengtsson and K.~Yczkowski,
\newblock On mutually unbiased bases,
\newblock {\em Int.~J.~Quantum Inf.} {\bf 8}(4) (2010) 535--640.

\bibitem{Kn96}
E.~Knill,
\newblock Group representations, error bases and quantum codes,
\newblock Technical Report LAUR-96-2807, Los Alamos National Laboratory (1996).

\bibitem{KlRo04}
A.~Klappenecker and M.~R\"otteler,
\newblock Constructions of mutually unbiased bases,
\newblock in {\em Finite fields and applications}, Lecture Notes in Comput.~Sci., Vol.~2948 (Springer, 2004) pp.~137--144.

\bibitem{AsChWo07}
M.~Aschbacher, A.M.~Childs and P.~Wocjan,
\newblock The limitations of nice mutually unbiased bases,
\newblock {\em J.~Algebraic Combin.} {\bf 25}(2) (2007) 111--123.

\bibitem{Vo07}
A.~Vourdas,
\newblock Quantum systems with finite Hilbert space: Galois fields in quantum mechanics,
\newblock {\em J.~Phys.~A: Math.~Theor.} {\bf 40}(33) (2007) R285--R331.

\bibitem{Vourdas08JFAA}
A.~Vourdas,
\newblock Harmonic analysis on a Galois field and its subfields,
\newblock {\em J.~Fourier Anal.~Appl.} {\bf 14} (2008) 102--123.

\bibitem{HAPS89}
G.B.~Folland,
\newblock {\em Harmonic analysis in phase space},
\newblock Annals of Mathematics Studies, No.~122 (Princeton University Press, Princeton, NJ, 1989).

\bibitem{AlDVTo09}
P.~Albini, E.~De Vito and A.~Toigo,
\newblock Quantum homodyne tomography as an informationally complete positive-operator-valued measure,
\newblock {\em J.~Phys.~A: Math.~Theor.} {\bf 42}(29) (2009) 295302.

\bibitem{LaPe10}
P.~Lahti and J.-P.~Pellonp\"a\"a,
\newblock On the complementarity of the quadrature observables,
\newblock {\em Found.~Phys.} {\bf 40}(9-10) (2010) 1419--1428.

\bibitem{KiSc13}
J.~Kiukas and J.~Schultz,
\newblock Informationally complete sets of Gaussian measurements,
\newblock {\em J.~Phys.~A: Math.~Theor.} {\bf 46}(48) (2013) 485303.

\bibitem{CaHeScTo14}
C.~Carmeli, T.~Heinosaari, J.~Schultz and A.~Toigo,
\newblock Nonuniqueness of phase retrieval for three fractional Fourier transforms,
\newblock in press on {\em Appl.~Comput.~Harmon.~Anal.} (2014), doi:10.1016/j.acha.2014.11.001.

\bibitem{SuTo07}
P.~\v{S}ulc and J.~Tolar,
\newblock Group theoretical construction of mutually unbiased bases in Hilbert spaces of prime dimensions,
\newblock {\em J.~Phys.~A: Math.~Theor.} {\bf 40}(50) (2007) 15099--15111.

\bibitem{ShVo11bis}
M.~Shalaby and A.~Vourdas,
\newblock Tomographically complete sets of orthonormal bases in finite systems,
\newblock {\em J.~Phys.~A: Math.~Theor.} {\bf 44}(34) (2011) 345303.

\bibitem{ApDaFu14}
D.M.~Appleby, H.B.~Dang and C.A.~Fuchs,
\newblock Symmetric Informationally-Complete Quantum States as Analogues to Orthonormal Bases and Minimum-Uncertainty States,
\newblock {\em Entropy} {\bf 16}(3) (2014) 1484--1492.

\bibitem{Isaacs06}
I.M.~Isaacs,
\newblock {\em Character theory of finite groups},
\newblock (AMS Chelsea Publishing, Providence, RI, 2006),
\newblock corrected reprint of the original {\em Character theory of finite groups}, Pure and Applied Mathematics, No.~69 (Academic Press, New York-London, 1976).

\bibitem{Mackey49SvN}
G.M.~Mackey,
\newblock A theorem of Stone and von Neumann,
\newblock {\em Duke Math.~J.} {\bf 16} (1949) 313--326.

\bibitem{Gu67}
A.~Guichardet,
\newblock {\em Le\c cons sur certaines alg\`ebres topologiques: Alg\`ebres de von Neumann; Alg\`ebres topologiques et fonctions holomorphes; Alg\`ebres de Banach commutatives},
\newblock (Gordon \& Breach, Paris-London-New York, 1967; distributed by Dunod Editeur).

\bibitem{GQT85}
V.S.~Varadarajan,
\newblock {\em Geometry of Quantum Theory}, 2nd edition,
\newblock (Springer-Verlag, New York, 1985).

\bibitem{Artin}
E.~Artin,
\newblock {\em Geometric algebra},
\newblock Wiley Classics Library (John Wiley \& Sons, Inc., New York, 1988),
\newblock reprint of the original {\em Geometric algebra} (Interscience Publishers, Inc., New York-London, 1957).

\bibitem{DeDM03}
J.~Dehaene and B.~De Moor,
\newblock Clifford group, stabilizer states, and linear and quadratic operations over $GF(2)$,
\newblock {\em Phys.~Rev.~A} {\bf 68}(4) (2003) 042318.

\bibitem{Appleby05JMP}
D.M.~Appleby,
\newblock Symmetric informationally complete-positive operator valued measures and the extended Clifford group,
\newblock {\em J.~Math.~Phys.} {\bf 46}(5) (2005) 052107.

\bibitem{Ap09}
D.M.~Appleby,
\newblock Properties of the extended Clifford group with applications to SIC-POVMs and MUBs,
\newblock arXiv:0909.5233 (2009).

\bibitem{BoRoWa61I}
B.~Bolt, T.G.~Room and G.E.~Wall,
\newblock On the Clifford collineation, transform and similarity groups. I,
\newblock {\em J.~Austral.~Math.~Soc.} {\bf 2} (1961) 60--79.

\bibitem{BoRoWa61II}
B.~Bolt, T.G.~Room and G.E.~Wall,
\newblock On the Clifford collineation, transform and similarity groups. II,
\newblock {\em J.~Austral.~Math.~Soc.} {\bf 2} (1961) 80--96.

\bibitem{Bl93}
L.~Blasco,
\newblock Paires duales r\'eductives en caract\'eristique $2$,
\newblock {\em M\'em.~Soc.~Math. France (N.S.)} {\bf 52} (1993) 1--73.

\bibitem{GuHa12}
S.~Gurevich and R.~Hadani,
\newblock The Weil representation in characteristic two,
\newblock {\em Adv.~Math.} {\bf 230}(3) (2012) 894--926.

\bibitem{Chau05}
H.F.~Chau,
\newblock Unconditionally secure key distribution in higher dimensions by depolarization,
\newblock {\em IEEE Trans.~Inform.~Theory} {\bf 51}(4) (2005) 1451--1468.

\bibitem{SuWo07}
W.K.~Wootters and D.M.~Sussman,
\newblock Discrete phase space and minimum-uncertainty states,
\newblock arXiv:0704.1277 (2007).

\bibitem{GuHaSo08}
S.~Gurevich, R.~Hadani and N.~Sochen,
\newblock The finite harmonic oscillator and its associated sequences,
\newblock {\em Proc.~Natl.~Acad.~Sci.~USA} {\bf 105}(29) (2008) 9869--9873.

\bibitem{GuHaSo08bis}
S.~Gurevich, R.~Hadani and N.~Sochen,
\newblock The finite harmonic oscillator and its applications to sequences, communication, and radar,
\newblock {\em IEEE Trans.~Inform.~Theory} {\bf 54}(9) (2008) 4239--4253.

\bibitem{Hump75}
J.E.~Humphreys,
\newblock {\em Linear algebraic groups},
\newblock Graduate Texts in Mathematics, No.~21 (Springer-Verlag, New York-Heidelberg, 1975).

\bibitem{Karp}
G.~Karpilovsky,
\newblock {\em The Schur multiplier},
\newblock London Mathematical Society Monographs New Series, No.~2 (The Clarendon Press, Oxford University Press, New York, 1987).

\bibitem{GuHa07}
S.~Gurevich and R.~Hadani,
\newblock The geometric Weil representation,
\newblock {\em Selecta Math.~(N.S.)} {\bf 13}(3) (2007) 465--481.

\bibitem{Klep65}
A.~Kleppner,
\newblock Multipliers on abelian groups,
\newblock {\em Math.~Ann.} {\bf 158} (1965) 11--34.

\bibitem{SeLe80}
G.~Seroussi and A.~Lempel,
\newblock Factorization of symmetric matrices and trace-orthogonal bases in finite fields,
\newblock {\em SIAM J.~Comput.} {\bf 9}(4) (1980) 758--767.

\bibitem{Ima83}
K.~Imamura,
\newblock On self-complementary bases of {${\rm GF}(q^n)$} over ${\rm GF}(q)$,
\newblock {\em Trans.~IECE Japan (Section E)} {\bf 66} (1983) 717--721.

\bibitem{JuMeVa90}
D.~Jungnickel, A.J.~Menezes and S.A.~Vanstone,
\newblock On the number of self-dual bases of ${\rm GF}(q^m)$ over ${\rm GF}(q)$,
\newblock {\em Proc.~Amer.~Math.~Soc.} {\bf 109}(1) (1990) 23--29.

\end{thebibliography}
\end{document}